\pdfoutput=1
\documentclass{aastex}
\usepackage{spr-astr-addons}
\usepackage{url}
\usepackage{widetext}
\usepackage{amsthm}
\usepackage{float}
\newtheorem{theorem}{Theorem}
\RequirePackage[colorlinks,citecolor=blue,linkcolor=blue]{hyperref}
\RequirePackage{color}

\newcommand{\emaila}{ h.shenavar@gmail.com, \\ h.shenavar@mail.um.ac.ir}

\begin{document}
\title{Motion of Particles in Solar and  Galactic Systems by Using  Neumann Boundary Condition }
\slugcomment{Not to appear in Nonlearned J., 45.}
\shorttitle{Orbits of Particles}
\shortauthors{Hossein Shenavar}

\author{Hossein Shenavar\altaffilmark{1}} 
\affil{ Department of Physics, Ferdowsi University of Mashhad,  Mashhad, Iran.\\ \emaila }

\begin{abstract}
A new equation of motion, which is derived previously by imposing Neumann boundary condition on cosmological perturbation equations   \citep{ShenavarI}, is investigated. By studying the precession of perihelion, it is shown that the new  equation of motion suggests a  small, though detectable, correction in orbits of solar system objects. Then a  system of particles is surveyed to have a better understanding of galactic structures. Also  the general form of the force law is introduced by which the rotation curve and mass discrepancy of axisymmetric disks of stars are derived.  In addition, it is suggested that the mass discrepancy as a function of centripetal acceleration becomes significant near a constant acceleration $ 2c_{1}a_{0} $ where $c_{1}$ is the Neumann constant and $ a_{0} = 6.59 \times 10^{-10} $ $m/s^{2}$ is a fundamental acceleration. Furthermore, it is shown that a critical surface density equal to $ \sigma_{0}=a_{0}/G $, in which G is the Newton gravitational constant, has a significant role in rotation curve and mass discrepancy plots. Also, the specific form of NFW mass density profile at small radii, $ \rho \propto 1/r $, is explained too.  Finally, the present model will be tested by  using a sample of 39 LSB galaxies for which we will show that the rotation curve fittings are generally acceptable. The derived mass to light ratios too are found within the  plausible bound except for the galaxy  F571-8.
\end{abstract}
 \keywords{ Neumann Boundary Condition; Constant Acceleration in EoM; Precession of perihelion; Pioneer Anomaly; Galactic Rotation Curve; Galactic Mass Discrepancy. }

\section{Introduction} 

The boundary condition of Einstein field equation has been debated since the begining of the general theory of relativity.  It has been shown that the Einstein-Hilbert action of general relativity is not well-posed in the sense that the boundary condition of the action is not compatible with the equations of field. See  \citet{Krishnan, Chakraborty} and the references therein. The root of the problem  is related to the fact  that Einstein field contains second-order derivatives of the metric \citep{Chakraborty}. However, one can overcome this obstacle by adding a boundary term to Einstein-Hilbert action to cancel the surface terms. In fact, as explained by \cite{Charap}, there could be infinitely many boundary terms to do this. The most famous boundary term is the YGH boundary term \citep{York, Gibbons} which makes action invariant under diffeomorphism \citep{Chakraborty}. In $ D $ dimension, we can write the YGH action which presumes Dirichlet boundary conditions, i.e. kills all the normal derivatives of the metric tensor on the surface,  as
\begin{equation}  \nonumber
S= \frac{1}{2\kappa}\int_{M} d^{D}x \sqrt{−g}(R - 2\Lambda)
+ \frac{1}{\kappa} \int_{\partial M} d^{D-1}y \sqrt{|h|} \epsilon K
\end{equation}
in which $ \kappa =8\pi G $, $R$ is the Ricci scalar, $\Lambda$ is the cosmological constant, $h$ is the induced metric on the boundary  $ \partial M $ with the coordinates $ y $, $ K$ is the extrinsic curvature of the boundary, $g$ is the metric and $\epsilon = +1$ for timelike and -1 for spacelike boundaries. See \citet{Krishnan, Chakraborty} for proof. Also it has been shown by \cite{Krishnan} that  the action under Neumann boundary condition could be written as:
\begin{equation}   \nonumber
S= \frac{1}{2\kappa}\int_{M} d^{D}x \sqrt{−g}(R - 2\Lambda)
+ \frac{D-4}{2 \kappa} \int_{\partial M} d^{D-1}y \sqrt{|h|} \epsilon K
\end{equation}
in which, surprisingly, the boundary term vanishes in four dimensions $ D=4 $. Based on this  action they have concluded that perhaps \textit{"standard Einstein-Hilbert gravity in four dimensions, without boundary terms, has an interpretation as a Neumann problem".} This unique feature of imposition of Neumann BC on GR field equations suggests that  this boundary condition might be more interesting than we have thought so far. 

In a recent paper  we imposed Neumann boundary condition to the equations of cosmological perturbation \citep{ShenavarI}.  As a result of  this new boundary condition, a modified Friedmann equation and a new lensing equation were found; where the latter was tested by a sample of ten strong lensing systems. In addition, we used the concept of geometrodynamic clocks  to modify the equation of motion of massive particles. We applied this equation to a sample containing  101 HSB and LSB galaxies and re-estimated the value of the Neumann constant which was found to be compatible with the prior evaluations from Friedmann and lensing equations. Moreover, by using a Newtonian approach we were able to derive the growth of structures in early universe and then we showed that the structures now grow  more rapidly in matter dominated era.

In this work, I apply the new equation of motion of massive particles to solar and galactic systems. Also I compare the results with CDM model, Milgrom's modified Newtonian dynamics and Mannheim's fourth order conformal gravity.  The CDM model tries to solve the apparent discrepancies between  the dynamics of galactic - and also extra galactic - systems  and their observed masses by assuming a dark halo around these systems  \citep{DMP,NFW1,NFW2}. Variety of obstacles like the so-called rotation curves of galaxies \citep{mondrc1,mondrc2},   the stability of galactic systems  through numerical simulations \citep{stable1,stable2} and structure formation  are solved by this assumption.   

 On the other hand, to solve the mentioned discrepancies, \citet{Milgrom1,Milgrom2,Milgrom3} considered the possibility that the inertia term of Newton's second law is not proportional to the acceleration of the object, but rather is a more general function of it as $m \mu (\frac{a}{a_{0}}) a = F$, in which $a $ is the acceleration, $a_{0}$ is a constant acceleration, $F$ is the force and  the function $ \mu $ - known as the interpolating function- plays a significant role at small accelerations (for example in galactic scales).  Modified Newtonian dynamics (MOND) can perfectly explain asymptotic flatness of rotation curves \citep{mondrc1,Sanders1,Sanders2,Sanders3}. Also it is shown that  the zero-point of the baryonic Tully-Fisher relation is of the order of  $ Ga_{0} $ which simply could  be derived from MOND equation of motion.  The same is true for the zero-point of the Faber-Jackson relation in isothermal systems like elliptical galaxies \citep{Milgrom4}. There is also another quantity, $ \sigma_{0}= a_{0}/G $, which plays a significant role in galactic scales. Disks with mean surface density less than $\sigma_{0} $ have  added stability \citep{Milgrom5}. In addition,  $ \sigma_{0} $ defines a transition central surface density and $ \sigma_{0}/2 \pi  $ is very close to the central surface density of dark halos \citep{Milgrom6}. Moreover, this model simply explains that features in the rotation curve should follow the baryonic distribution; so the Renzo's rule which states that "For any feature in the luminosity profile there is a corresponding feature in the rotation curve" \citep{Sancisi,McGaugh1}.  Also, through this model, it seems that mass discrepancy is related to acceleration in units of $ a_{0} $ \citep{McGaugh1,Lelli}. The main point is that, after so many data fittings of experts working on MOND model, the role of a constant  acceleration, $ a_{0} $, in galactic scales seems irrefutable now \citep{mond}. 

Another proposal to solve the dark matter problem is the  fourth-order conformal theory of gravity \citep{mannheim1,mannheim2,mannheim3,mannheim4}, which is related to the present model  because this theory too, predicts a constant acceleration.  The field equation of this theory is derived from the action $I_{w} = -\alpha \int d^{4} x \sqrt{-g} C_{\lambda \mu \nu \kappa} C^{\lambda \mu \nu \kappa}$ where $ C_{\lambda \mu \nu \kappa}  $ is the conformal Weyl tensor and $ \alpha $ is a purely dimensionless coefficient. This action leads to a fourth-order field equation which has an exact vacuum solution of the form  $ds^{2}= B(r)dt^{2} - \frac{dr^{2}}{B(r)} - r^{2} d\Omega $ in which $ B(r) = 1 - \frac{a}{r} + b + cr + \lambda r^{2}$. As it is clear, the fourth term on the rhs of $ B(r) $ establishes a constant acceleration which plays a critical role in this theory. It has been shown that by using this added linear potential, one could capture the general trend of the rotation curve data to a good degree without needing any dark matter. In addition, obtained mass to light ratios, $ M/L $, are close to the values of the local solar neighborhood \citep{mannheim6,mannheim7}. Conformal gravity claims that mass discrepancy in galactic scales is due to a global cosmological effect on local galactic motions. The cosmological constant problem too could be suppressed if we include the amount by which it gravitates \citep{mannheim5}.

In the present work,  I start by surveying the motion of a single particle subjected to the new equation of motion $a = g - 2c_{1}cH(t)$ derived in \cite{ShenavarI}. In this equation, $a $ is the acceleration, $ g$ is the gravitational field, $ H(t)=\dot{R(t)}/R(t) $ is the Hubble parameter, $c$ is the speed of light and $c_{1}$ is the Neumann constant. We also use $a_{0}= cH(t)$ to clarify the similarities with MOND dynamics, though, the values are not quite the same. I will also study the effects of the new term $2c_{1}cH(t)$ in the scales of the solar system  and I derive the precession of perihelion for some planets and dwarf planets. In the third section, the mentioned equation is applied to a system of particles and most importantly a modified force and potential is derived which is used in the next section to find the rotation curve of a disk galaxy. In the fifth section we study the functional dependence of mass discrepancy on radius, acceleration and gravitational field. Also we will argue that the present model can justify the radial dependence of NFW mass profile of dark halos and its constant surface density. Eventually, we use a data  sample, that includes 39 LSB galaxies, to  test the predictions of our model and then conclude our discussion with some final remarks.

\section{The Motion of a Single Particle }
We first consider the motion of a particle $ m $ in a gravitational field. As we discussed before, the modified equation of motion of this particle in a centrally directed gravitational field is as follows \citep{ShenavarI}:
\begin{equation}    \label{main}
\frac{d^{2}\vec{r}}{dt^{2}} + 2c_{1}a_{0}\hat{e_{r}} = g(r) \hat{e_{r}}
\end{equation} 
where $ \vec{r} = r \hat{e} $ is the position of the particle, $ a_{0}= cH_{0} $ is a constant acceleration, also observed in \citet{Milgrom1,Milgrom2,Milgrom3} modified dynamics with a different value, and $ g(r) $ is the Newtonian field of gravity. Here we can put the new term $ 2c_{1}a_{0}\hat{e_{r}} $ on the lhs of the equation of motion and  treat it as a modification to dynamics   or ,when we transfer  $ 2c_{1}a_{0}\hat{e_{r}} $ to the rhs of Eq. \eqref{main}, we can think of this term  as a new force. These two approaches will eventually lead to the same results. For example, it is possible to show that both of these two scenarios result in the following conservation of energy $E$:
\begin{equation}   \label{energy}
1/2m \vec{v}^{2} +2mc_{1}a_{0}r +V(\vec{r})= E
\end{equation}
where $ V(\vec{r}) = -\int^{\vec{r}} _{\vec{r_{0}}} m \vec{dx} .\vec{g}(\vec{x})$ is the gravitational potential energy. To prove this,  one may  consider $ \vec{F} =m(g(\vec{r})-2c_{1}a_{0})\hat{e_{r}} $ as a new force and use the definition of work done against this force in moving a particle from an initial place $ \vec{r_{0}} $ to a final place $ \vec{r} $:
\begin{equation}
W(\vec{r}) = \int^{\vec{r}} _{\vec{r_{0}}}  \vec{dx} .\vec{F}
\end{equation}
Then it is possible to derive the conservation of energy \eqref{energy} because the acting force is conservative. Another way to obtain the conservation of energy  is to directly integrate Eq. \eqref{main} to derive conservation of energy Eq. \eqref{energy}. As we mentioned, these two approaches are equivalent; however, we prefer to use the concepts of force and potential because the central force problem, which we need here, could be find in any standard textbook of theoretical mechanics and so   it is customary to follow this approach.

From Eq. \eqref{main} one can prove that $ \frac{d}{dt}(\vec{r} \times \frac{d\vec{r}}{dt}) = 0 $ which is expected because the  model is spherically symmetric. Therefore, similiar to any spherically symmetric model,  the particle here moves in a plane which is known as the orbital plane. Now, by using plane polar coordinate $ (r,\psi) $ and define the Lagrangian of the motion as follows:
\begin{equation}
\mathcal{L} = \frac{1}{2}m(\dot{r}^{2} +(r\dot{\psi})^{2}) - \Phi,
\end{equation}   
in which $ \Phi = \Phi_{N} +2mc_{1}a_{0}r $ is the total potential and $\Phi_{N}$ is the Newtonian potential energy, it is possible to  derive  the equations of motion from Euler-Lagrange equation \citep{Goldstein}:
\begin{equation} \label{Radial}
 m\ddot{r} -mr\dot{\psi}^{2} +\frac{d\Phi}{dr}=0   
\end{equation}
\begin{equation}  \label{Angular}
\frac{d}{dt}(mr^{2}\dot{\psi}) =0
\end{equation}
The second equation shows that the quantity $ l=mr^{2}\dot{\psi} $ is another constant of the motion which is known as the angular momentum.  Another proof of the conservation of angular momentum $\vec{l}= m \vec{r}\times \vec{v}$ could be obtained by considering its rate $d\vec{l}/dt$, and using these facts that the external torques are negligible and we deal with a central force here.

The radial equation of motion \eqref{Radial} will contain only $ r $ and its derivatives if we replace  $ \dot{\psi}  $  by $ l/mr^{2} $, . Therefore, this equation will be equivalent to the equation of  a one dimensional motion in which a particle is subjected to an effective potential of:
\begin{equation}
\Psi = \Phi_{N} + 2c_{1}m a_{0}r +\frac{l^{2}}{2mr^{2}}
\end{equation} 
where the third term on the rhs is due to the familiar centrifugal force. From Eq. \eqref{energy} one can rewrite the conservation of energy  as:
\begin{equation}    \label{Energy1}
E = \Psi +\frac{1}{2}m \dot{r} ^{2}
\end{equation}
Now we will study this one dimensional model for the specific case of an attractive inverse-square law of force, i.e. the Kepler problem $ \Phi_{N}=-\frac{k}{r} $ where a positive $ k $ describes a force toward the center. First consider the escape velocity of a particle which is defined as the minimum velocity required to escape from the gravitational field.  According to the discussion that we had before \citep{ShenavarI}, if the particle  reaches the radius $ R_{0} \approx \frac{2c_{1}a_{0}}{\Lambda} \approx 100 Mpc $, in which $ \Lambda $ is the cosmological constant, then it is free. The reason is that in this radius the repulsive force due to cosmological constant, $ F_{1} \approx m \Lambda r  $, overcomes the total attractive force $ F =m(g(r)-2c_{1}a_{0}) $.  Although for most physical cases,  one could neglect $g(r)$ at a distance like $ R_{0} $ because $ g(R_{0}) \ll 2c_{1} a_{0} $. Thus from Eq. \eqref{Energy1} the escape velocity is approximately $ v_{0}^{2} \approx 2c_{1}a_{0}R_{0} $  and every particle with an energy larger than $ v_{0}^{2} $ could come from infinity, strike the repulsive centrifugal barrier, be repelled and then travel back to infinity. These are very energetic particles. 
 On the other hand all particles with smaller velocities than $ v_{0} $ are bounded to the central object. For example consider a particle with energy $ E_{1} $ in Fig. \ref{fig:potential}. For this particle there are two turning points, $ r_{1} $ and $ r_{2} $, also known as "apsidal distances" \citep{Goldstein}. According to Bertrand's theorem, which states that the only central forces that result in closed orbits are the inverse-square law and Hooke's law \citep{Goldstein}, orbits of the present model are not closed because we have an added constant force too. However,  if energy of the particle coincides with the minimum of the effective potential then $ r_{1}=r_{2} $ and the orbit is a circle and thus closed. 

\begin{figure}[!hbp]
\centering
\includegraphics[height=4cm,width=8cm]{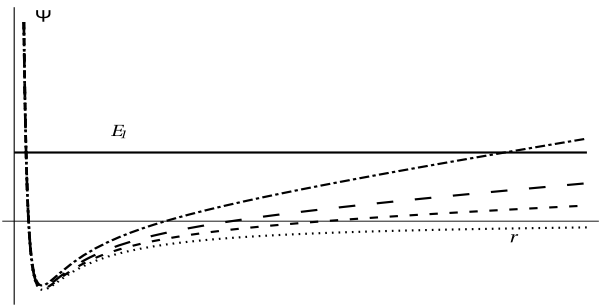}    \\
  \caption{General outline of the  effective potential $ \Psi $ as a function of radius $r$. A particle with total energy $E_{1}$ would have two turning points.}
  \label{fig:potential}
\end{figure}

Another important feature of Eq. \eqref{main} is that for this new equation of motion the famous Laplace-Runge-Lens  (LRL) vector is not a constant of the motion anymore. We will show this below, however we should mention that in the Kepler problem beside four independent constants of the motion, i.e. the three elements of angular momentum vector and the energy, LRL vector is another constant which always points in the same direction known as the line of apsides \citep{Goldstein}. In the present model though, the cross product of Eq. \eqref{main} with the constant angular momentum $ \vec{l}$  results in (after a little manipulation):
\begin{equation}
\frac{d}{dt}(\vec{p} \times \vec{l} -mk\frac{\vec{r}}{r}) = 2c_{1}m^{2}a_{0}(r \dot{\vec{r}}  -\dot{r}\vec{r})
\end{equation} 
where $ \vec{p} $ is momentum of the particle, $ k $ depends on the mass of the particle and the mass of the source of the gravity while $ \vec{A} = \vec{p} \times \vec{l} -mk\frac{\vec{r}}{r} $ is the LRL vector. Although $ \vec{A} $ is not a constant anymore, which means that the line of apsides changes direction, a natural question is whether there is any other constant of the motion. The answer, which is no,  is connected to the nonclosed nature of the orbits. See \cite{Goldstein} page $ 105 $.  As we discussed before, due to Bertrand's theorem, the orbitals of Eq. \eqref{main} are nonclosed; therefore, as $ \psi $ goes around the particle never retraces its footsteps on any previous orbit. Thus $ r $ is an infinite-valued function of $ \psi $ and so the additional conserved quantity of the motion should involve an infinite-valued function of $ \psi $ too. So there is no more simple-defined  constant of the motion.

Before we turn to the problem of a system of particles, which is needed to describe galactic phenomena, we  consider the perturbation effect of $ a_{0} $-term in Eq. \eqref{main} on the outer objects of the solar system.  To do so, we use $ \eta = 2c_{1}a_{0}/g_{N} $ as a measure of the strength of the new term in \eqref{main} relative to the Newtonian gravitational field $ g_{N}  $. Using this parameter,  near the surface of the Sun we have $ \eta \approx 10^{-14}  $, the Earth $ \eta \approx 10^{-9} $, Neptune $ \eta \approx 10^{-6}  $ and Eris (one of the last dwarf planets) $ \eta \approx 10^{-5}  $. Therefore the new term is very small compared to the Newtonian  gravity term in the solar system. 

\begin{table*}[t]
\caption{Derived precession of perihelion $\bar{\dot{\omega}}$ for two planets and two dwarf planets  in arc seconds per century   $~^{\prime\prime} / cy$. Data of semi major axis $A$ in $AU$, eccentricity $e$, orbital period $\tau$ (year) are derived from NASA http://nssdc.gsfc.nasa.gov/.} 
\centering  
\begin{tabular}{c c c c c c c c} 
\hline\hline                        
Object & A (AU) & e  & $\tau$ (y) & C & rev/cent & $\bar{\dot{\omega}}_{GR}$ ($~^{\prime\prime} / cy$) & $\bar{\dot{\omega}}_{a_{0}}$ ($~^{\prime\prime} / cy$)     \\ [0.5ex] 
\hline                  
\text{URANUS} & 19.2 & 0.046 & 83.7 & 6.323 & 1.195 & $2.4 \times 10^{-3}$ & 21.128 \\
 \text{NEPTUNE} & 30.05 & 0.011 & 163.7 & 6.286 & 0.611 & $7.8 \times 10^{-4}$ & 41.333 \\
 \text{PLUTO} & 39.48 & 0.245 & 247.9 & 7.558 & 0.403 & $4.2 \times 10^{-4}$  & 50.052 \\
 \text{ERIS} & 67.781 & 0.433 & 558 & 11.546 & 0.179 & $1.2 \times 10^{-4} $& 74.843 \\ [1ex]        
\hline 
\end{tabular}
\label{table:Perihelia} 
\end{table*}[t]

\begin{table*}[t]    
\caption{Derived precession of perihelion $\bar{\dot{\omega}}$  for four inner planets and the asteroid Icarus in arc seconds per century  $~^{\prime\prime} / cy$. Data of semi major axis $A$ in $AU$, eccentricity $e$, orbital period $\tau$ (year) are derived from NASA http://nssdc.gsfc.nasa.gov/. The observed values of perihelion precession are reported in \cite{Inverno} page 198 except for Mars which is reported in \cite{ohanian}. All results of $\bar{\dot{\omega}}_{a_{0}}$ are within the accepted bounds of the data. } 
\centering  
\begin{tabular}{c c c c c c c c c c} 
\hline\hline
Object &  A (AU) & e  & $\tau$ (y) & C & $A/A_{Neptune} $ & rev/cent & $\bar{\dot{\omega}}_{obs}$ ($~^{\prime\prime} / cy$) & $\bar{\dot{\omega}}_{GR}$ ($~^{\prime\prime} / cy$) & $\bar{\dot{\omega}}_{a_{0}}$ ($~^{\prime\prime} / cy$) \\ [0.5ex] 
\hline               
 \text{MERCURY} & 0.387 & 0.205 & 0.241 & 7.157 & 0.0129 & 414.9 & 43.1$\pm $0.5 & 42.96 & 0.06 \\
 \text{VENUS} & 0.723 & 0.007 & 0.615 & 6.284 & 0.0241 & 162.6 & 8.4$\pm $4.8 & 8.63 & 0.15 \\
 \text{EARTH} & 1. & 0.017 & 1        & 6.288 & 0.0333 & 100 & 5.0$\pm $1.2 & 3.84 & 0.25 \\
 \text{MARS} & 1.52 & 0.094 & 1.88 & 6.453 & 0.0506 & 53.2 & 1.32$\pm$? & 1.35 & 0.47 \\
 \text{ICARUS} & 1.078 & 0.837 & 1.12 & 149.913 & 0.0359 & 89.3 & 9.8$\pm $0.8 & 10.04 & 0.66 \\
\hline 
\end{tabular}
\label{table:Perihelia2} 
\end{table*}

Although the expected effects are very small, it is still possible to find some clues of the new term in the scale of solar system. For example, one can use canonical perturbation theory to evaluate precession rate averaged over a period of unperturbed orbit $ \tau $ as \citep{Goldstein}:
\begin{equation}     \label{Precession}
\bar{\dot{\omega}} =  \frac{\partial \overline{\Delta H}}{\partial l}
\end{equation}    
where $ \Delta H =2c_{1}ma_{0}r $ is the perturbation Hamiltonian and 
\begin{equation}
\overline{\Delta H} = \frac{2c_{1}ma_{0}}{\tau} \int^{\tau}_{0} r dt
\end{equation}
is the time average of the perturbation. Then by using the conservation of angular momentum, i.e. $ ldt=mr^{2}d \psi $, the last integral converts to:
\begin{equation}
\overline{\Delta H} = \frac{2c_{1}m^{2} a_{0}}{l\tau} (\frac{l^{2}}{mk})^{3} \int_{0}^{2\pi} \frac{d\psi}{[1+e \cos(\psi - \psi^{\prime})]^{3}}
\end{equation} 
in which $ e $ is the eccentricity and $ \psi^{\prime} $ is one of the turning angles  of the orbit.  It is also  possible to evaluate the precession rate from the average rate of LRL vector $(1/ | \vec{A} |) d\vec{A}/dt$; though we prefer to use the standard method of \citep{Goldstein}. Anyway, one can calculate the integral in the last equation for any given eccentricity $e$; the result, say $ C $, is most likely of the order of $ 10 $. See Table \ref{table:Perihelia} and \ref{table:Perihelia2} for the magnitude of $ C$ for some important objects of the solar system. Then, from Eq. \eqref{Precession} the averaged precession rate is as follows:
\begin{equation}     \label{Precession1}
\bar{\dot{\omega}}_{a_{0}} =  \frac{8c_{1}C}{\tau}\frac{ A^{2} a_{0}}{GM} (1-e^{2})^{2} 
\end{equation} 
where $ M $ is the mass of the Sun and we have used the relation $ l^{2}= mkA(1-e^{2}) $ which is correct for an unperturbed orbit with a semimajor axis $ A $ \citep{Goldstein}.  In fact, \cite{Sultana} have found the perihelion shift of a test particle in fourth-order conformal gravity by tracking timelike geodesics which is compatible
with the present result. Of course, one expects this agreement because both models are mainly based on linear potential terms.

We should point out that Eq. \eqref{Precession1} is independent of the mass of the planet and so it represents a shear geometrical effect. The same is true about the conventional precession equation of GR. However, , it should be mentioned that  we expect a geometrical precession rate in both models, because the underlying field equation, i.e. Einstein's field equations, are geometrical; though, here we have used a different boundary condition as discussed in the Introduction. In addition, from Eq. \eqref{Precession1} it is clear that orbits with larger semimajor axis will show larger precession rate. Therefore this effect should be more apparent in outer parts of the solar system. See Table \ref{table:Perihelia} for derived perihelion precession per century of two planets -Uranus and Neptune- and two dwarf planets - Pluto and Eris. In this table we have assumed that $ a_{0} = 6.59 \times 10^{-10} $ $m/s^{2}$.

As we know, general relativity predicts another correction to Newtonian motion that can be described by a potential proportional to $ 1/r^{3} $. This term comes from Schwarzschild solution of Einstein field equation in strong field limit around the Sun \citep{ohanian}. By using the same method as we applied here, one can evaluate precession rate averaged over a period for the GR case as $ \bar{\dot{\omega}}_{GR} = \frac{6\pi}{\tau (1-e^{2})}(\frac{GM}{c^{2}A}) $. See \cite{Goldstein} for full derivation. For the planet Mercury we have  $ \bar{\dot{\omega_{1}}} \approx 0.10^{\prime \prime} $ per revolution or $ 42.96^{\prime \prime} $ per century.  Table \ref{table:Perihelia2} demonstrates our model's prediction about inner planets. In fact, to have a reliable model of precission rate of inner plants, we should include the effects of all planets on inner ones. Therefore, full revisiting of the potential theory of the solar system is needed. This development is beyond the scope of the present work; thus, here we model the solar system by a uniformly distributed sphere of mass to roughly evaluate the effects of planets. As we will show in Appendix \ref{app} - after we surveyed systems of particles in the next section - the acceleration at radius $r$ of a uniformly distributed sphere by radius $R_{sphere}$ is proportional to $2c_{1}a_{0}\frac{r}{R_{sphere}}$. Approximating the radius of the solar system by the semimajor axis of Neptune $R_{sphere} = A_{Neptune} = 30 AU$, we estimate the  effect of the other  planets on the inner ones, by multiplying a factor of $ r/A_{Neptune} $ to Eq. \eqref{Precession1}. Therefore, for example in the Earth and Mercury we put $2c_{1}a_{0} \frac{1 AU}{30 AU} $ and $ 2c_{1}a_{0}\frac{0.387 AU}{30 AU} $ respectively, instead of $2c_{1} a_{0} $  to obtain the precession rate.

The  derived precession rate of inner planets are reported in Table \ref{table:Perihelia2}. In the case of Mercury, we find  $\bar{\dot{\omega}}_{a_{0}}=0.06^{\prime\prime} / cy$ (arc seconds per century) and  $\bar{\dot{\omega}}_{GR}=42.96^{\prime\prime} / cy$ while the observed value -subtracting the major effect of the Newtonian effects of the other planets - is equal to $43.1 \pm $ $0.5^{\prime\prime} / cy$ which shows that $\bar{\dot{\omega}}_{a_{0}} $ stays within the observational bound. The same is true about Venus and Earth and asteroid Icarus as you may see in Table \ref{Precession1}. In the case of Mars, the observational precission rate with its uncertainty could not be found, though, by investigating the uncertainties of the other planets shows that it should be around $\pm 0.5$ or more. Therefore, Mars too is within our model.   It is important to emphasize that although the effect of the new term $2c_{1}a_{0}$ is small compared to the GR effect, and very tiny compared to the Newtonian theory, it is not negligible at all. A thorough study  of $ \bar{\dot{\omega}}_{a_{0}} $ of inner planets and also asteroids  could help to evaluate its reliability, though, the best objects to search about such possible effects are the outer planets and the dwarf planets.

Because for orbits with  large eccentricities, the value of $ C $ grow rapidly,  we expect  larger precession rates. For example, for $ e=0.8 $ and  $ e=0.9 $ we have $ C=106.65 $ and $ C=561.01 $ respectively. The asteroid Icarus for example, which we talked about in the last paragraph, has an eccentricity of $e=0.837$ and so $C= 149.913$. High eccentricities are also common  for comets. This could be a helpful clue in our future investigations. Consider, for now, the Halley's comet with eccentricity of $ e=0.967 $ and semimajor axis of  $17.94 AU $.  This comet makes a perfect case for the reason that it has a very high eccentricity and one can derive $ C=8589 $ and $\bar{\dot{\omega}}_{a_{0}}=  108.96^{\prime\prime} / cy $. Although it should be mentioned that the factor $ (1-e^{2})^{2}$ is significant here too because it decreases the precession for large eccentricities. On the other hand, consider the comet  Chiron with $ A=13.7 AU$ and $ e= 0.383$. Data are again from NASA. For this comet one can find    $ C=10.02 $ and $\bar{\dot{\omega}}_{a_{0}}=   14.69^{\prime\prime} / cy $. Thus the precession rate  in the orbit of Halley's comet should be larger by almost an order of magnitude than the precession rate  in the orbit of Chiron. The question is whether this new term could explain some peculiar features of  the  motion of comets in the solar system.  However, this issue  is beyond the scope of our work because in the case of comets we have some close encounters with Jupiter and Saturn. Also the possibility of nongravitational forces should be included too \citep{Newburn}.
The point here is that if Eq. \eqref{main} is supposed to be correct, one should be able to detect its effects at the scales of the solar system.

At the scales of the solar system, one can find  some other independent evidence too. For example, tracking data of Pioneer 10 and 11 spacecraft have shown a systematic unmodeled acceleration about $ (8.74 \pm 1.33)\times 10^{-10} m/s^{2} $ directed toward the Sun \citep{Anderson1,Anderson2,Anderson3}. If there is indeed such added acceleration as Eq. \eqref{main}  at solar scale, it must have some effects on small objects of the this system too. In fact, \cite{Wallin}  considered the idea of an added constant acceleration, selecting a well-observed sample of trans-Neptunian objects with orbits between 20 and 100 AU  from the Sun, and placed tight bound on the magnitude of the constant acceleration. According to this research, the deviation from inverse square law of gravity should be about $ 8.7 \times 10^{-11} ~ m/s^{2}$ to $ 1.6 \times10^{-10}~ m/s^{2} $. This estimation supports the value of the added acceleration of the present model, i.e. $2c_{1}a_{0} = 8.6\times 10^{-11}$ though it is more close to the lower bound.

Even in larger scales of the solar system  like the Oort cloud,  \cite{Iorio} has  shown that a constant acceleration toward the Sun could make bound trajectories and these trajectories  radically differ from those of the Newtonian theory. Therefore, we see that there are some clues that support  the existence of a constant acceleration  in the scale of solar system and therefore, it is now safe to test Eq. \eqref{main} in larger scales like galaxies and clusters of galaxies. To do so, we need to investigate the issue of a system of particles all interacting according to Eq. \eqref{main}.

Before we discuss systems of particles, we should point out alternative models too have some predictions about precession rate in solar system. For instance, \cite{Gron} assumed a constant density of dark matter in solar system in hydrostatic equilibrium. Then, based on the uncertanties in the perihelion precession of the asteroid Icarus, which is about $8 \%$, they have shown that the density of DM would be about $\rho_{0} \leqq  1.8 \times 10^{-16} g/cm^{3} $ which is about 7 orders of magnitude larger than average galactic mass density. Therefore, we should conclude that either the uncertanty in the data could not be atributed to DM or the current data of the perihelion precession do not put strict limits on the density of it.

Fourth-order conformal gravity, on the other hand, has succeeded in matching data and predictions in the case of perihelion precession. \cite{Sultana} have used data for perihelion shift observations and found constraints on the
value of the fundamental constant of fourth-order gravity which coincides with that obtained by using galactic rotational curves. 
 
\section{Systems of Particles}
In this section the mechanics of a system of particles will be studied. To do so, we need some new definitions. Suppose that $ \vec{F}^{e}_{i} $ is the external force on the $ i $th particle, with mass $ m_{i} $, and let  $ \vec{F}^{N}_{ij}=\frac{Gm_{1}m_{2}}{r^{2}_{ij}}\hat{e_{ij}} $ be the Newtonian force exerted on the $ i $th particle by the $ j $th particle ( with $r_{ij}$ being the distance between the two particles and unit vector $\hat{e_{ij}}$ pointing from i to j ) and finally $ \vec{F}^{a_{0}}_{ij}$ be the term due to acceleration transformation. Clearly one has $ F^{N}_{ii} = F^{a_{0}}_{ii} =0 $ and also  $ \vec{F}^{N}_{ij} = - \vec{F}^{N}_{ji} $. On the other hand, because $ \vec{F}^{a_{0}}_{ij} $ is proportional to a constant acceleration we have:
\begin{equation} \label{forces}
\frac{F^{a_{0}}_{ij}}{F^{a_{0}}_{ji}} = \frac{m_{i}}{m_{j}}
\end{equation}
Using these definitions, the equation of motion of the $ i $th particle under the influence of external and internal forces is as follows:
\begin{equation}
m \ddot{\vec{r}}_{i} =  \vec{F}^{e}_{i}  + \sum_{j}( \vec{F}^{N}_{ij} + \vec{F}^{a_{0}}_{ij})
\end{equation}
while, if the last equation is summed for all other particles, it takes the form of:
\begin{equation}    \label{System}
\sum_{i} m_{i} \ddot{\vec{r}}_{i} = \sum_{i} \vec{F}^{e}_{i}  + \sum_{ij}( \vec{F}^{N}_{ij} + \vec{F}^{a_{0}}_{ij})
\end{equation}

Now it is possible to introduce four useful quantities: the center of mass or weighted average of the radii vectors of particles $ \vec{R} = 1/M \sum m_{i}\vec{r}_{i} $, total linear momentum of the particles $ \vec{P} = \sum m_{i} \frac{d\vec{r_{i}}}{dt} $, summation of radii $ \vec{\varrho} = \sum \vec{r}_{i} $ and summation of velocities $ \vec{\Pi} = \sum \frac{d \vec{r}_{i}}{dt} $. We use these definitions  to prove the following important theorem:
\begin{theorem}  
a) If the total external force is zero and $ \vec{F}^{a_{0}}$ is negligible, then the total linear momentum $ \vec{P} $ is conserved,

 b) If the total external force is zero and the Newtonian gravitational force $ \vec{F}^{N}$ is negligible, then the summation of velocity $ \vec{\Pi} $ is conserved,
 
  c) If the total external force is equal to zero, but neither Newtonian gravitational force nor the force due to constant acceleration $ \vec{F}^{a_{0}}$ are negligible, then neither  $ \vec{P} $ nor $ \vec{\Pi} $ are conserved. 
\end{theorem}   
 
\begin{proof}
Part a). Under the conditions of part (a) which refer to strong field limit, the rhs of \eqref{System} is equal to zero since the law of action and reaction states that each pair $ \vec{F}^{N}_{ij} + \vec{F}^{N}_{ji} $ is zero. Thus the total linear momentum is conserved. Part b).  In this case, which refers to weak field limit, it is easier to start from \eqref{main} and sum over all particles. Then it follows:
\begin{equation}    \label{System1}
\sum_{i} \ddot{\vec{r}}_{i} = \sum_{i} \frac{\vec{F}^{e}_{i}}{m_{i}}  + \sum_{ij}( \dfrac{\vec{F}^{N}_{ij}}{m_{i}} + 2c_{1}a_{0}\hat{e_{ij}})
\end{equation}
For  every two particles $i$ and $j$, the last terms on the rhs of this equation $ 2c_{1}a_{0}\hat{e_{ij}} $ are equal and opposite and  so cancel each other. Now, if the total external force $ \vec{F}^{e} $ is zero and the Newtonian gravitational force $ \vec{F}^{N}$ is negligible, then it is clear that the rhs of Eq. \eqref{System1} is equal to zero and so $ \vec{\Pi}=\sum \frac{d \vec{r}_{i}}{dt} $ is a constant of the motion. Part c). This case is clear from Eqs. \eqref{System} and \eqref{System1} because the rhs of none of these equations would vanish if neither Newtonian gravitational force nor the force due to constant acceleration $ \vec{F}^{a_{0}}$ are negligible. However there is an important exception here: if all particles in a system have the same mass then both  $ \vec{P} $ and $ \vec{\Pi} $ are conserved.  
\end{proof} 

As an example, let us consider the important case of a two-body  system. Suppose that the total external force is negligible, then by adding equations of motions of these two particles, it is possible to find equation of motion of the center-of-mass:
\begin{equation}
M\ddot{\vec{X}} = 2c_{1}a_{0}\hat{e_{12}} (m_{2}-m_{1})
\end{equation} 
where $ \vec{X}=1/M(m_{1}\vec{x}_{1}+m_{2}\vec{x}_{2}) $ is the center of mass position and $ M = m_{1}+m_{2}$ is the total mass of the system. Therefore the conservation of the total linear momentum of the system is violated unless the two particles have same masses $ m_{2} = m_{1} $ as we mentioned in the last paragragh. To derive the relative motion of $ m_{2} $ and $ m_{1} $ one should multiply the equation of motion of the former by $ m_{2}/M $ and the latter by $ m_{1}/M $ and then subtract to achieve:
\begin{equation}
\mu \ddot{\vec{x}} = \vec{F}^{N}_{12} + 4\mu c_{1} a_{0}\hat{e_{12}}
\end{equation}
where $ \vec{x} = \vec{x}_{1}-\vec{x}_{2} $ is the relative position of the two objects and $ \mu =\frac{m_{1}m_{2}}{M} $ is known as the reduced mass. As one expects, for the case that the two particles have the same mass $m$ we obtain $ m \ddot{\vec{x}} = \vec{F}^{N}_{12} + 2c_{1}m a_{0}\hat{e_{12}} $, which is exactly of the form of equation of motion of single particle under the influence of gravitational field $ \vec{F}^{N}_{12} $ and a constant acceleration.

To obtain the net force on a particle $ m_{i} $  at position $ \vec{x} $ from a system of particles, one should simply add the small contribution from each small particle. However, for a typical galaxy by about $ 10^{11} $ stars this method is not practicable. A good solution to this problem is to assume a smooth mass distribution  which is  proportional to the local star density $ \rho(\vec{x^{\prime}}) $ at any point. According to Eq. \eqref{main}, the tiny force exerted by a small element of mass ( of a large mass distribution  ) $ dm^{\prime}=\rho(\vec{x^{\prime}}) d^{3}\vec{x^{\prime}} $ positioned at $ \vec{x^{\prime}} $ on a point particle $ m_{i} $ placed at $\vec{x}$ is as follows:
\begin{equation}
\delta \vec{F}(\vec{x}) = Gm_{i}\frac{\vec{x^{\prime}} - \vec{x}}{| \vec{x^{\prime}} - \vec{x}|^{3}} \rho(\vec{x^{\prime}})d^{3}\vec{x^{\prime}} + 2c_{1}a_{0} m_{i}\frac{\vec{x^{\prime}} - \vec{x}}{| \vec{x^{\prime}} - \vec{x}|}
\end{equation}  
Then the total force of the mentioned distribution on the point particle $ m_{i} $ can be deduced by summing these  small contributions:
\begin{eqnarray}   \nonumber   \label{Force}
\vec{F}(\vec{x}) = m_{i} (  G \int \frac{\vec{x^{\prime}} - \vec{x}}{| \vec{x^{\prime}} - \vec{x}|^{3}}\rho(\vec{x^{\prime}})d^{3}\vec{x^{\prime}}  \\   
     + \frac{2c_{1}a_{0}}{M} \int \frac{\vec{x^{\prime}} - \vec{x}}{| \vec{x^{\prime}} - \vec{x}|} \rho(\vec{x^{\prime}})d^{3}\vec{x^{\prime}} ). 
\end{eqnarray}
Here $ M $ is the total mass of the distribution. To obtain the second term on the rhs of \eqref{Force} we used the fact that the acceleration due to the new term $ 2c_{1}a_{0} $ between $ m_{i} $ and the mass distribution  should be the same. On the other hand, as Eq. \eqref{forces} states, the ratio of the forces between these two parts is equal to the ratio of their masses and so we can obtain the last equation. 

An interesting example to apply Eq. \eqref{Force} to is the force between a spherical shell of matter,  by mass $ M $, and a particle of mass $ m $ inside or outside of the shell. For this problem , the first integral on the rhs of Eq. \eqref{Force} refers to the famous shell theorems: {\it A uniform shell exerts no gravitational force on a particle inside the shell; the shell attracts an external particle as if all the mass of the shell were concentrated at its center} . For the second integral on the rhs of Eq. \eqref{Force} it is straightforward to show that for an external particle the exerted force is $ ma_{0}(1-R^{2}/3r^{2}) $  while for an internal particle we find $ 2ma_{0}r/3R  $. Here $ R$ represents the radius of the shell and $ r $ is the  distance of the particle from shell' s center.  See appendix \ref{app} for the derivation. Deriving the force for other symmetrical mass distributions - i.e. homogeneous sphere, an exponential sphere and an exponential disk - shows that for an internal particle the leading term of the force is always proportional to $r/R$.  This is the reason that in the previous section we demanded to use $ 2c_{1}a_{0}\frac{r}{R} $ instead of $ 2c_{1}a_{0} $ in the case of inner planets like the earth.    Furtheremore, it is easy to check that the new force is always toward the center of the shell. In addition, from the symmetry, it is clear that there is no net force on a particle located at $ r=0 $.  We will use these results to have a better understanding of the solutions and their asymptotic behaviour.

As usual, we prefer to work with scalar quantities instead of vector ones. Having this in mind, we can define the  potential $\Phi(\vec{x})$ of a distribution of matter by $\Phi(\vec{x})= \int \vec{F}(\vec{x^{\prime}}).d\vec{x^{\prime}} $ and derive it as:
\begin{eqnarray}   \label{Potential}
\Phi=-G \int \frac{\rho(\vec{x^{\prime}})d^{3}\vec{x^{\prime}}}{| \vec{x^{\prime}} - \vec{x}|}  + \frac{2c_{1}a_{0}}{M}\int \rho(\vec{x^{\prime}})d^{3}\vec{x^{\prime}}|\vec{x^{\prime}} - \vec{x}|
\end{eqnarray} 
where the first term is the Newtonian gravitational potential and the second one is the new  potential due to the constant acceleration.  This new term provides extra attracting force as we need in galactic scales.

It is clear that the Newtonian gravitational potential  fulfills the Poisson equation;  though the net potential $ \Phi $ can not fulfill any second-order Poisson equation because of the presence of the new $a_{0}$ dependent potential $ \Phi^{a_{0}}=\frac{2c_{1}a_{0}}{M}\int \rho(\vec{x^{\prime}})d^{3}\vec{x^{\prime}}|\vec{x^{\prime}} - \vec{x}| $. It is possible, however, to show that the new potential $ \Phi $ satisfies a new fourth order Poisson equation:
\begin{equation}
\nabla^{4} \Phi = 4 \pi G \nabla^{2} \rho_{b} - \frac{16c_{1} \pi a_{0}}{M} \rho_{b}.
\end{equation}  
in which $ \rho_{b} $ is the baryonic matter. It should be pointed out that higher-order Poisson equations are very common in beyond Einstein gravity. For example, in fourth-order conformal gravity or $F(R)$ gravities, we essentially deal with an action with higher-than-second order derivatives. Therefore, higher order Poisson equations emerge quite naturally without even considering systems of particles. In these theories one should deal with the problem of ghosts, i.e. negative norm states which break the unitarity. However, in the case of fourth-order gravity, it has been shown that the ghost states disappear  from the eigenspectrum of the Hamiltonian   \citep{mannheim5}.

The situation in MOND is more complicated. In this model, a few modified Poisson equations have been proposed by using TeVeS or based on other related assumptions. See \cite{mond} for a complete review. This model is basically nonlinear, though some linearized versions of its Poisson equations has been reported \cite{mond}. The nonlinearity poses serious problems in using this great model.

The Poisson equation of the present model, i.e. the last equation,  is derived from a second-order action, i.e. Einstein-Hilbert action plus a boundary term. Thus it displays no ghost problem. Also it is genuinely linear and so more easy to use.  We discuss this new Poisson equation in more detail in a forthcoming paper in which we survey the distribution of matter in galaxies \citep{ShenavarIII}. Now we investigate the motion of particles in disk galaxies which is the main aim of the present work.

\section{Rotation Curve of Spiral Galaxies}
According to Newton's equation of motion, it is expected that the rotation velocity of particles orbiting  around a massive object falls  as $ v^{2} \sim \frac{1}{r} $ in which $r$ is the radius of the particle. This behavior is known as the Keplerian fall-off of velocity and it is highly supported by the solar system data. However, at galactic scales such decline in the velocity has not been observed. In fact, because of the hydrogen gas which is distributed in galaxies to much further distances than their stars, we now have detailed knowledge about the rotation velocity of objects at galactic scales and it is now clear that there is a mass discrepancy at least for galaxies which have  extended rotation curves beyond their optical radius. See \cite{DMP} for an excellent history.

According to  \cite{Casertano2}, the rotation curve data of galaxies could be divided into three categories. First, the rotation curve of low luminosity dwarf galaxies, with maximum velocity smaller than $ 100 $km/s, are typically rising. Second, those of intermediate to high luminosity galaxies, with $ 100 < v_{max} < 180 $ km/s and $ v_{max}> 180 $ km/s respectively, are found to be flat. Third, the rotation curves of the very highest luminosity galaxies are generally declining from $ 15 \% $ for NGC2903  to $ 30 \% $ for NGC2683.  

In is known that galaxies have very diverse mass distributions; however, for the sake of simplicity we consider  the special case of disk galaxies here. The disk galaxies contain stars, gas and dust. There are also  spiral structures with various shapes and lengths.  We usually measure distribution of stars in a galactic disk by observing total stellar luminosity emitted by unit area of the disk, i.e the surface brightness. Observations suggest that the surface brightness of these galaxies is approximately an exponential function of radius: $I(r) = I \exp(-r/R_{d})$; where $ R $ is the radius and $ R_{d} $ is the disk scale length.  Scale lengths range from $ 1~  kpc $ to more than $ 10 ~  kpc $; see  \citep{Binney} page 26. While  galactic disks are thin because mass density falls off much faster perpendicular to the equatorial plane than in the radial direction,  it is believed that most spiral galaxies contain also a bulge, which is a centrally concentrated stellar system. Here we seek a thin disk rotation curve; the extension to a separable thick disk or  a disk with a spherical bulge is straightforward.

\begin{figure}[!hbp]
\centering
\includegraphics[height=4cm,width=8cm]{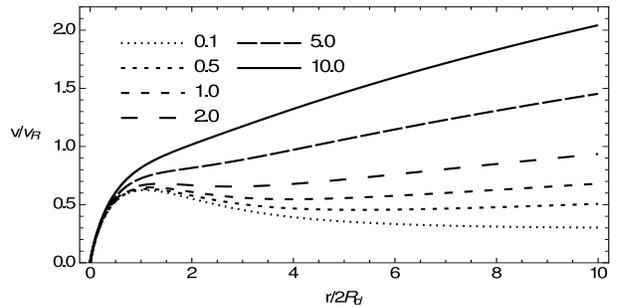}    \\
  \caption{Rotation velocity $v$ in units of $ v_{R}=\sqrt{GM/R_{d}} $ as a function of radius for different values of $ \sigma_{0} / \Sigma_{0} = 0.1$, $ 0.5$, $1.0$, $2.0$, $5.0$, $10.0 $. Low mass galaxies, i.e. galaxies with small $\Sigma$, show a clear rising rotation curve; those of intermediate- to high-mass galaxies have a relatively flat rotation curve, and those of the highest masses posses a Keplerian decline in their rotation curves to a large distance. }
  \label{fig:rotationcurve}
\end{figure}
 
If we suppose that, beside surface brightness of these galaxies, the disk surface mass density $\Sigma(R)$  is also exponential
\begin{equation} 
\Sigma(R) = \Sigma_{0} \exp(-r/R_{d}),
\end{equation}
in which $ \Sigma_{0} $ is the  surface mass density at $r=0$, then the total mass of the disk is $ M=2\pi \Sigma_{0} R^{2}_{d} $. Following an approach which \cite{Casertano1} has developed, by differentiating Eq. \eqref{Potential} with respect to $ r$ one could derive the circular velocity of the exponential disk:
\begin{eqnarray}  \nonumber   \label{Rcurve}
v^{2} (r) = r \frac{\partial \Phi}{\partial r} = ~~~~~~~~~~~~~~~~~~~~~~~~~~~~~~~~~~~~~~~~~~ \\ \nonumber
   4\pi \Sigma_{0} GR_{d} y^{2} [I_{0}(y)K_{0}(y)- I_{1}(y)K_{1}(y)]  ~~~~~   \\
       + 8c_{1}a_{0}R_{d}y^{2}  I_{1}(y)K_{1}(y)  ~~~~~~~~~~~~~~~~~~~~~~~~~
\end{eqnarray} 
in which $ I $ and $ K $ are modified Bessel functions and $ y=r/2R_{d} $. Complete derivation of the last equation is presented in \ref{app}. See also  \citet{mannheim3,mannheim5} for the same derivation related to another linear potential. One could also find the potential of a separable thick disk and  a spherical bulge in the latter reference. If the structure is composed of two or more elements, for example a cold exponential disk plus some HI gas distribution with a specific profile, then one has to evaluate rotation curve for each element and  add them up to derive the total rotation curve. For example, \cite{mannheim7} have approximated the gas profile  as single exponential disk with scale length equal to four times of those of the corresponding optical disk. This is a very successful approximation as fitted rotation curves clearly show \citep{mannheim7}.

\begin{figure}[!hbp]
\centering
\includegraphics[height=4cm,width=8cm]{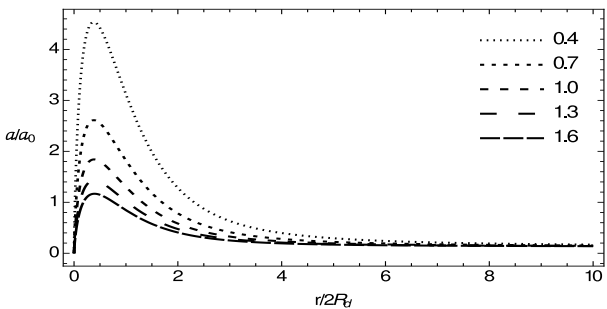}    \\
  \caption{Centripetal acceleration $ a $ in units of $ a_{0} $ as a function of radius for different values of $ \Sigma_{0}/\sigma_{0} = 0.4$, $ 0.7$, $1.0$, $1.3$, $1.6 $. Regardless of the mass, size and shape of galaxies the present model predicts that the centripetal acceleration should approach  a constant value of $ 2c_{1}a_{0} $ while Newtonian theory of gravity predicts a decreasing centripetal acceleration as $ 1/r^{2} $. See  \cite{ShenavarI} for a sample of galaxies which clearly suggests a constant acceleration at their last data points.}
  \label{fig:acceleration5}
\end{figure}

All departure  from the standard Newtonian mechanics are embodied in the $ a_{0} $-dependent term of the Eq. \eqref{Rcurve}. This term is competitive with the Newtonian one in outer parts of galaxies and clearly it predicts  that the more luminous the galaxy, the less important $ a_{0} $ term is; or the less hypothetical dark matter is needed. To see this, let  us rewrite the last equation as a dimensionless form by  dividing it by the maximum velocity of an equivalent spherical distribution $ v^{2}_{R}= GM/R_{d} $:
\begin{eqnarray}   \label{Rcurve1}
\frac{v^{2} (r)}{GM/R_{d}} =  2y^{2} [I_{0}(y)K_{0}(y)- I_{1}(y)K_{1}(y)]  \\ \nonumber
       + \frac{4c_{1}}{\pi}\frac{\sigma_{0}}{\Sigma_{0}}y^{2}  I_{1}(y)K_{1}(y)  ~~~~~~~~~~~~~~~~~~
\end{eqnarray}
where $ \sigma_{0} = a_{0}/G $ is a critical mass density. See Fig. \ref{fig:rotationcurve} for a plot of this equation for various ratios of $\sigma_{0}/ \Sigma_{0}$. Regarding Eq. \eqref{Rcurve1} it is evident that in galaxies with large surface mass density, i.e. large $ \Sigma_{0} $, the second term is less important while it would be dominant for galaxies with small surface density. In general, we see that the parameter $\sigma_{0}/ \Sigma_{0}$  determines the general behavior of the rotation curve here.

Detailed rotation curve data for a sample of 39 galaxies will be presented in the last section; here, we seek theoretical clues of Eq. \eqref{main} for future investigations. Another important feature of this equation is that it suggests a constant acceleration, i.e. $ 2c_{1}a_{0}$, at large radii from the center of the disk.   See Fig. \ref{fig:acceleration5} for a plot of the acceleration $ a =v^{2}/r$ scaled by $a_{0}$ as a function of scaled radius $ r/2R_{d} $ for different values of $\sigma_{0}/ \Sigma_{0}$. Outside of a disk, the acceleration as a function of radius could be approximated as $ GM/r^{2} + 2c_{1}a_{0}-3c_{1}a_{0}R^{2}_{d}/r^{2} $ which  we have used asymptotic behavior of modified Bessel functions to derive the last result. See Appendix \ref{app} for the formula for asymptotic behavior of Bessel functions. Therefore, it is clear that, although more massive galaxies reach a higher peak in their acceleration plot, the data of all galaxies  should asymptotically converge to  $ 2c_{1}a_{0} $ eventually. In our last paper \citep{ShenavarI} we showed that for a sample of galaxies the accelerations of the last data points are close to this predicted value.  This behavior is also strongly supported by Mannheim and O'Brien investigations \citep{mannheim6,mannheim7}. They have studied two different galaxy samples consisting of high surface brightness (HSB), low surface brightness (LSB) and dwarf galaxies with different masses, scales and rotation velocities. In the first sample they have considered 111 HSB, LSB and dwarf spiral galaxies. The second sample  is consisted of 27 objects which most of them are dwarf galaxies. The data are such that the value of the quantity $ (v^{2}/r)_{last} \approx 3 \times 10^{-11} m/s^{2}$ as measured at the last data point is near universal.  This quantity is also very close to the constant acceleration of the conformal Weyl gravity  which is numerically extracted from the data fitting of rotation curves \citep{mannheim5,mannheim6,mannheim7}. In addition, MOND theory \citep{Milgrom1} and Moffat's metric skew-tensor theory \citep{Brownstein}, which have succeeded to explain galactic rotation curve, possess two similar  universal parameters. For Mond theory one has $ a_{0} \approx 10^{-10} m/s^{2} $ and for MSTG theory $ G_{0}M_{0}/r^{2}_{0}c^{2} \approx 7.67 \times 10^{-29} cm^{-1}$ \citep{mannheim5,Brownstein}. Even more support for a role of a universal acceleration could be found  as we will discuss in the next section \citep{McGaugh1,Lelli}. Therefore, it seems that the existence of a constant quantity should be regarded as an important empirical clue for galactic systems.

\section{Mass Discrepancy}

In any alternative theory of gravity it is important to find new quantities to illustrate differences between new theory, the Newtonian one and the CDM paradigm. One of these quantities is  the mass discrepancy which is  defined as the squared ratio of observed velocity $v$ to the velocity that is attributable to Newtinian gravity of visible baryonic matter $ v_{N}$, viz. $(\frac{v}{v_{N}})^2$ \citep{McGaugh1,mond}. It has been shown that this quantity is well correlated with acceleration, and increases systematically with decreasing acceleration below the  MOND parameter  $ 1.2 \times  10^{-10} m/s^{2}$ \citep{McGaugh1}.  In the present model it is easy to start from equation \eqref{Rcurve1} and derive the mass discrepancy  as a function of radius r, centripetal acceleration $ a= \frac{v^2}{r} $ and Newtonian acceleration $ g_{N}$ as follows:
\begin{equation}     \label{massdis1}
(\frac{v}{v_{N}})^2 = 1 + 2c_{1}\frac{\sigma_{0}}{\Sigma_{0}}\frac{I_{1}(y)K_{1}(y)}{I_{0}(y)K_{0}(y)- I_{1}(y)K_{1}(y)}
\end{equation}              

\begin{equation}    \label{massdis2}
(\frac{v}{v_{N}})^2 = \frac{a}{a-4c_{1}a_{0}yI_{1}(y)K_{1}(y)}
\end{equation}

\begin{equation} \label{massdis3}
(\frac{v}{v_{N}})^2 = 1 + \frac{4c_{1}a_{0}}{g_{N}}yI_{1}(y)K_{1}(y)
\end{equation} 
Let us investigate the implications of these three equations. Equation \eqref{massdis1} predicts that in galaxies which surface mass densities are much smaller than $ \sigma_{0} $, mass discrepancy should exist everywhere. On the other hand, galaxies with a high surface density should show a larger mass discrepancy in their outer parts but not in their inner parts.  Thus, mass discrepancy as a function of radius depends on the ratio $ \sigma_{0} / \Sigma_{0} $ for all galaxies. See Fig. \ref{fig:massdiscr} for a plot of mass discrepancy as a function of radius for different values of $ \sigma_{0} / \Sigma_{0} $ ratios. Also you may find  data plots of mass discrepancy for various galaxies in  \cite{McGaugh1}. We should also point out that  if there is a bulge at the center of the galaxy then one should expect a higher mass discrepancy at small radii; though, the mass discrepancy at large radii would not change significantly.
\begin{figure}[!hbp]
\centering
\includegraphics[height=4cm,width=8cm]{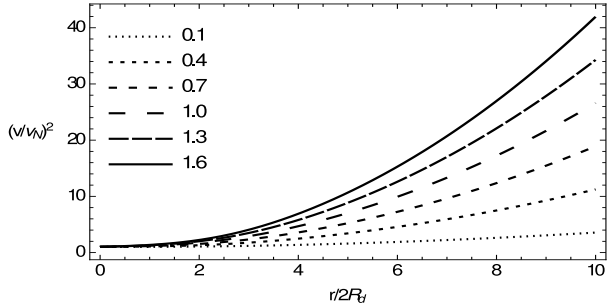}    \\
  \caption{Mass discrepancy as a function of radius for different values of $ \Sigma_{0}/\sigma_{0} =0.1$, $0.4$, $ 0.7$, $1.0$, $1.3$, $1.6 $. For larger values of $ \sigma_{0} / \Sigma_{0} $ the mass discrepancy is larger compared to the ones with smaller values of $ \sigma_{0} / \Sigma_{0} $. }
  \label{fig:massdiscr}
\end{figure}

\cite{McGaugh1} have shown that, regardless of the choice of stellar mass-to-light ratio, acceleration is the physical scale with which the mass discrepancy correlate best. See also \cite{mond}. On the one hand, equation \eqref{massdis2} illustrates that for any galaxy when centripetal acceleration approaches to $ 2c_{1}a_{0}$, we should see a large mass discrepancy. See figure \ref{fig:massdisca}. In addition, the special form of \eqref{massdis2} seems generic for spiral galaxies; i.e. if equation \eqref{main} is correct, all of spiral galaxies should have the same relation between mass discrepancy and centripetal acceleration. This is an important point because  the previous results, viz. Eqs. \eqref{Rcurve} and \eqref{massdis1}, seems to vary from galaxy to galaxy because they depend on the central surface density $ \Sigma_{0} $; however equation \eqref{massdis2} shows a unique character for all spiral galaxies. On the other hand, when centripetal acceleration is much larger than $ 2c_{1}a_{0} $, for example at scales of the  solar system, mass discrepancy is very small. In these limits the mass discrepancy approaches to one. See the horizontal asymptote of Fig. \ref{fig:massdisca}. In addition, for any radius there is a vertical asymptote at $ a \lesssim 2c_{1}a_{0} $. Furthermore, from equation \eqref{massdis3} we see a large mass discrepancy when $ g_{N} $ approaches to zero. However when $ g_{N} $ is much larger than $2c_{1}a_{0}$, mass discrepancy is very small. To conclude, we should say that if equation \eqref{main} is ought to be  the correct description of object's motion in galactic scales, equations \eqref{massdis1}, \eqref{massdis2} and \eqref{massdis3} should be applicable in any such scales, except around the centers of galaxies which we might need to use the strong field regime of the general theory of relativity. 

\begin{figure}[!hbp]
\centering
\includegraphics[height=4cm,width=8cm]{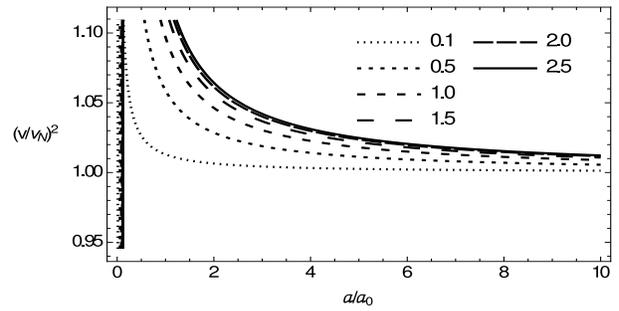}    \\
  \caption{Mass discrepancy $(\frac{v}{v_{b}})^2 $ as a function of centripetal acceleration $a$ for different values of $ y=r/R_{d} =0.1 $, $0.5  $, $1.0 $, $1.5 $, $2.0 $, $2.5$. The horizontal line, $(\frac{v}{v_{b}})^2 = 1 $ , indicates where there is no mass discrepancy; i.e. baryonic mass content is enough to explain the observed motion. It is expected, according to our model, that the mass discrepancy becomes important when $ a<4c_{1} a_{0} $.}
  \label{fig:massdisca}
\end{figure}

\textbf{ One last point:} Comparing the squared of the observed velocity $ v $ and Newtonian velocity $ v_{N} $, this time with their subtraction instead of their ratio,  could make a useful result. From Eq. \eqref{Rcurve} it is easy to see that:
\begin{equation}
\frac{v^{2}-v^{2}_{N}}{r}=4c_{1}a_{0}yI_{1}(y)K_{1}(y)
\end{equation}
where $ v $ and $ r $ are observable quantities and $ v_{N} $ could be derived \citep{McGaugh1}. Data plot of this last equation could directly disprove Eq. \eqref{main} if we find inconsistent results. On the other hand it is possible to define at any radius $ r $, a dynamical mass which is attributable to the rotation velocity: $ M_{D}=\frac{rv^{2}}{G} $. Then, starting from the last equation, one may obtain the difference between dynamical and Newtonian mass as follows: $M_{D} - M_{N} = 16c_{1} \sigma_{0} R_{d}^{2} y^{3}I_{1}(y)K_{1}(y) $; where $ M_{N} $ is the amount of Newtonian - or baryonic - mass inside radius $ r $. Now we interpret the quantity $M_{DM}= M_{D} - M_{N} $ as the amount of "missing mass" or the dark matter mass, where clearly is related to the critical surface density $ \sigma_{0} =a_{0}/G $ but not to the surface density $ \Sigma_{0} $ of  galaxies. Therefore the dark matter density distribution varies, very interestingly, as ($ \rho_{DM} =\frac{dM_{DM}}{dV} $):
\begin{equation}   \label{DMD}
\rho_{DM}=\frac{c_{1}\sigma_{0}}{2 \pi R_{d}} \lbrace I_{1}(y)K_{1}(y)+y(I_{0}(y)K_{1}(y)-I_{1}(y)K_{0}(y))\rbrace
\end{equation}
in which we have assumed a spherical shape for the dark halo, viz. $dV = 4 \pi r^{2}dr$. See Appendix \ref{app}. The first interesting  point is that according to this equation the density of the dark halo should vary proportional to the inverse of radius  at small radii, i.e.  $ \rho_{DM}\approx (c_{1}\sigma_{0}/\pi)[r^{-1} +O(r^{-1})] $. We used Eqs. \eqref{Knu} and \eqref{Inu} in Appendix \ref{app} to derive this asymptotic behavior. This behavior  is the same as the behavior of the NFW mass profile   for the galactic dark halos at small radii. Navarro, Frenk and White (NFW) used high-resolution N-body simulations to study density profiles of dark matter halos and found $ \rho^{NFW}_{DM} = \frac{\rho_{0}}{(r/r_{s})(1+(r/r_{s}))^{2}}$ \citep{NFW1,NFW2}. They realized that all such profiles have the same shape, specially independent of the halo mass. Their specific mass profile changes gradually from $ 1/r $ to $ 1/r^{3} $ beyond a certain critical distance from the center known as the scale radius $ r_{s} $. This special form of mass distribution reproduces galactic phenomena; however it provides an infinite mass for the dark matter halo \citep{Binney}. The second important point about Eq. \eqref{DMD} is that it predicts a constant central surface density for dark halos $\frac{c_{1}\sigma_{0}}{ \pi } $, i.e. the central surface density should be independent of the total mass of galaxies. Interestingly, this is a famous observational result that was first noticed by  \cite{Donato}. See also \cite{mond}. Then \cite{Milgrom6} proved that MOND too predicts such  central surface density for the dark halos as $ a_{0, MOND}/(2 \pi G) = 0.286 $ $kg/m^{2}$   which is close to our derived value of $\frac{c_{1}\sigma_{0}}{ \pi } = 0.204 $ $kg/m^{2}$.  Therefore, we see that the present model explains the two important features of the NFW mass profile namely, i.e. its small radii behavior and its constant central surface density.

The point that the halo mass profile   is dependent to the critical mass density $\sigma_{0}$, nor  the surface density of single galaxies $\Sigma_{0}$, is a  particularly interesting result. The reason is that we usually relate the observational quantities to the baryonic mass; e.g. rotation curves of galaxies and mass discrepancies as discussed here or Tully–Fisher relation which will be studied elsewhere \citep{ShenavarIII}. Although, the mass profile that we predicted above, i.e. $\rho_{DM}$, is solely related to the halo and is independent of the baryonic mass. The fact that the present model can  predict a property of a totally different paradigm, i.e. the DM hypothesis, without even referring to the observational data seems quite notable. 

\begin{figure}[!hbp]
\centering
\includegraphics[height=4cm,width=8cm]{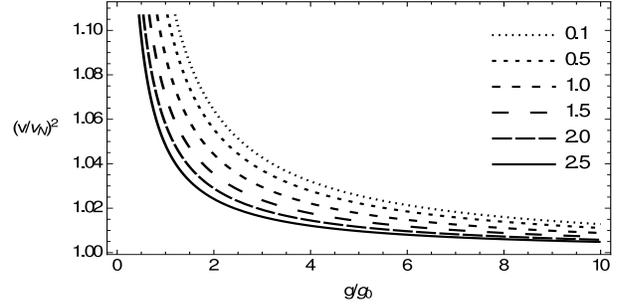}    \\
  \caption{Mass discrepancy $(\frac{v}{v_{b}})^2 $ as a function of Newtonian gravitational field $g$ for different values of $ y=r/R_{d} =0.1 $, $0.5  $, $1.0 $, $1.5 $, $2.0 $, $2.5$. }
  \label{fig:massdiscg}
\end{figure} 
 
 Eqs. \eqref{Rcurve} and \eqref{DMD} could explain another important feature of our model. From these two equations it is clear that all departure from pure Newtonian theory  are included in the $ a_{0} $-dependent terms. Therefore, because both of the functions are solely dependent to baryonic matter, any feature in rotation curve data and dark matter profile should occur simultaneously.  This is known as the Renzo's rule; i.e. "For any feature in the luminosity profile there is a corresponding feature in the rotation curve" \citep{Sancisi,McGaugh1}. In fact any pure baryonic theory could explain this observational rule.  It is important to note that all these results are achieved without assuming the dark matter hypothesis.  \\\\

\section{Data Analysis of 39 LSB galaxies}
In this section we apply the model that we discussed above to a sample of 39 LSB galaxies.   The names of the galaxies could be found in the first column of Table \ref{sample}. The distance D, luminosity $L_{B}$  and scale length $R_{d}$ of each galaxy with the sources of the data  are shown too. In this work, the presented distances of NED ( NASA/IPAC Extragalactic Database), reported by \cite{mannheim6}, are used which are based on Cepheids, Tully-Fisher relation or redshift measurement.  
 The estimated mass and mass to light ratio $M/L$, derived  from the curve fitting, are reported in the table too. In this data analysis, the galactic mass of  galaxies and the parameter $c_{1}$ are the free parameters of the fit. Following a discussion based on population synthesis model by \cite{Sanders1} and \cite{mannheim6}, we restricted the mass to light ratio to be larger than $0.2M_{\odot}/L_{\odot}$.  Also we found that the best fitting results are obtained when we restrict $c_{1}$ to be within the interval of $0.2 \times 0.065 \leq c_{1} \leq 2.0 \times 0.065$ which seems to be  reasonable due to our prior estimation of this parameter \citep{ShenavarI}. However, I should mention that the mean value that we found for $c_{1}$ in this data analysis appears to be about half of the predicted one $c_{1}=0.065$.  In \cite{ShenavarI} we used last data points of galactic rotation curves and found that the mean value of $c_{1}$ in the case of LSB galaxies is generally smaller compared to that of HSB galaxies. The difference compared to $c_{1}=0.065$ was found about 7 $\%$ lower for LSBs while while 37 $\%$ higher for HSBs. Comparing to these prior results, the mean value of $c_{1}$ that we found here, i.e. $c_{1} \approx 0.03$ needs more attention. The mean value is smaller, as we expected because we are dealing with LSBs; though it is half of the theoretical value. This might be due to our method of fitting or general uncertainties in the data. To find an accurate value for   $c_{1}$, we are urged to seek better fitting methods and more accurate data in future.

 \begin{table*}[t]    
  \centering
  \begin{tabular}{|c| c | c | c| c| c| c| c| c| c| c| c| c| c| c| }
    \hline  \hline
 \tiny Name & \tiny  D (Mpc) & \tiny $R_{d}$ (kpc) & \tiny Luminosity $L_{B}$ ($10^{10}L_{\odot})$ & \tiny $M_{B}$ ($10^{10}M_{\odot})$&\tiny Derived $M/L$ &  $c_{1}$ & RC Data Source \\    \hline
\text{DDO0064} & 6.8 & 1.3 & 0.015 & 0.003 & 0.2 & 0.043 & NMB  NMBB \\
 \text{ESO0140040} & 217.8 & 10.1 & 7.169 & 54.086 & 7.544 & 0.013 & MRB \\
 \text{ESO0840411} & 82.4 & 3.5 & 0.287 & 0.057 & 0.2 & 0.013  & MRB \\
 \text{ESO1200211} & 15.2 & 2. & 0.028 & 0.060 & 2.163 & 0.013 & MRB \\
 \text{ESO1870510} & 16.8 & 2.1 & 0.054 & 0.177 & 3.287 & 0.013 & MRB \\
 \text{ESO2060140} & 59.6 & 5.1 & 0.735 & 4.208 & 5.725 & 0.027 & MRB \\
 \text{ESO3020120} & 70.9 & 3.4 & 0.717 & 0.809 & 1.128 & 0.013 & MRB\\
 \text{ESO3050090} & 13.2 & 1.3 & 0.186 & 0.037 & 0.2 & 0.013 & MRB\\
 \text{ESO4250180} & 88.3 & 7.3 & 2.6 & 0.52 & 0.2 & 0.076 &MRB\\
 \text{E48800490} & 64.6 & 6.9 & 0.139 & 1.292 & 9.292 & 0.013 & MRB \\
 \text{F563-1} & 46.8 & 2.9 & 0.14 & 1.426 & 10.184 & 0.029 & MRB NMB  NMBB\\
 \text{F563-V2} & 57.8 & 2. & 0.266 & 0.639 & 2.403 & 0.013 & MRB NMB  NMBB\\
 \text{F568-3} & 80. & 4.2 & 0.351 & 0.616 & 1.756 & 0.026 & MRB NMB  NMBB\\
 \text{F571-8} & 50.3 & 5.4 & 0.191 & 5.897 & 30.872 & 0.013 & MRB \\
 \text{F579-V1} & 86.9 & 5.2 & 0.557 & 6.109 & 10.969 & 0.013 & MRB \\
 \text{F583-1} & 32.4 & 1.6 & 0.064 & 0.024 & 0.370 & 0.032 & MRB NMB NMBB\\
 \text{F583-4} & 50.8 & 2.8 & 0.096 & 0.319 & 3.325 & 0.013  & MRB NMB  NMBB\\
 \text{F730-V1} & 148.3 & 5.8 & 0.756 & 9.163 & 12.121 & 0.013 & MRB\\
 \text{NGC4395} & 4.1 & 2.7 & 0.374 & 1.824 & 4.877 & 0.013 &NMB  NMBB\\
 \text{NGC7137} & 25. & 1.7 & 0.959 & 0.349 & 0.364 & 0.013& NMB  NMBB\\
 \text{NGC959} & 13.5 & 1.3 & 0.333 & 0.352 & 1.056 & 0.040 & NMB  NMBB\\
 \text{UGC11454} & 93.9 & 3.4 & 0.456 & 3.205 & 7.029 & 0.017  & MRB\\
 \text{UGC11557} & 23.7 & 3. & 1.806 & 0.361 & 0.2 & 0.013  & MRB\\
 \text{UGC11583} & 7.1 & 0.7 & 0.012 & 0.002 & 0.2 & 0.016 & MRB\\
 \text{UGC11616} & 74.9 & 3.1 & 2.159 & 2.512 & 1.163 & 0.018 & MRB \\
 \text{UGC11648} & 49. & 4. & 4.073 & 4.387 & 1.077 & 0.013 & MRB \\         \text{UGC11748} & 75.3 & 2.6 & 23.93 & 9.042 & 0.378 & 0.043 & MRB\\
\text{UGC11819} & 61.5 & 4.7 & 2.155 & 4.638 & 2.152 & 0.040  & MRB\\
\text{UGC4115} & 5.5 & 0.3 & 0.004 & 0.0008 & 0.2 & 0.017 & MRB BMR \\
\text{UGC5750} & 56.1 & 3.3 & 0.472 & 0.094 & 0.2 & 0.013 & MRB BMR\\
\text{UGC6614} & 86.2 & 8.2 & 2.109 & 12.729 & 6.035 & 0.017 & MRB BMR\\
\text{UGC11820} & 17.1 & 3.6 & 0.169 & 2.108 & 12.474 & 0.13 & NMB  NMBB\\
\text{UGC1281} & 5.1 & 1.6 & 0.017 & 0.0469 & 2.760 & 0.014  & NMB  NMBB \\
\text{UGC128} & 64.6 & 6.9 & 0.597 & 5.703 & 9.553 & 0.014 & NMB  NMBB \\
\text{UGC1551} & 35.6 & 4.2 & 0.78 & 0.156 & 0.2 & 0.027 & BMR MRB \\
 \text{UGC191} & 15.9 & 1.7 & 0.129 & 0.170 & 1.31454 & 0.13 & NMB  NMBB \\
 \text{UGC4325} & 11.9 & 1.9 & 0.373 & 0.074 & 0.2 & 0.097 & NMB  NMBB \\
 \text{UGC477} & 35.8 & 3.5 & 0.871 & 1.289 & 1.480 & 0.013& NMB  NMBB \\
 \text{UGC5750} & 56.1 & 3.3 & 0.472 & 0.094 & 0.2 & 0.014 & NMB  NMBB \\  \hline
  \end{tabular}
  \caption{Properties of 39 LSB galaxies. Data of distance of galaxies D (Mpc), the scale length of galaxies $R_{d}$ (kpc) and their luminosity $L_{B}$ ($10^{10}L_{\odot})$ is derived from \cite{mannheim6}. The rotation curve data is extracted from \cite{Blok} shown as  BMR, \cite{McGaugh2} shown as MRB, \cite{Naray1} shown as NMB and shown as \cite{Naray2} NMBB. We have derived the mass  $M_{B}$ ($10^{10}M_{\odot})$ and mass to light ratio $M/L$.  }
\label{sample}
\end{table*}

Another point is that  no bulge is assumed in any of the cases. This presumption certainly affects the viability of the fitting in the inner parts of galaxies, though, it is clear that it should not strongly change  the rotation velocity at outer sections of  galaxies. The reason is that whatever new components we assume for galaxies, their contributions to the acceleration will converge to $2c_{1}a_{0}$ at large radii and thus it is independent of the shape of the new components. However, to obtain a better fitting result in inner parts, we intend to include the bulge contribution in future works.    

Beside this, the vertical thickness of the disk is ignored too. Galactic disk scale lengths $R_{d}$ are usually much larger than the  galactic scale heights. It can be proved that  the corrections due to thickness are only important in the inner parts of  galaxies, and therefore have no significant effect on the outer parts which is dominant by the linear potential contribution \citep{mannheim5}.

\begin{figure*}
\centering
\includegraphics[height=22cm,width=17cm]{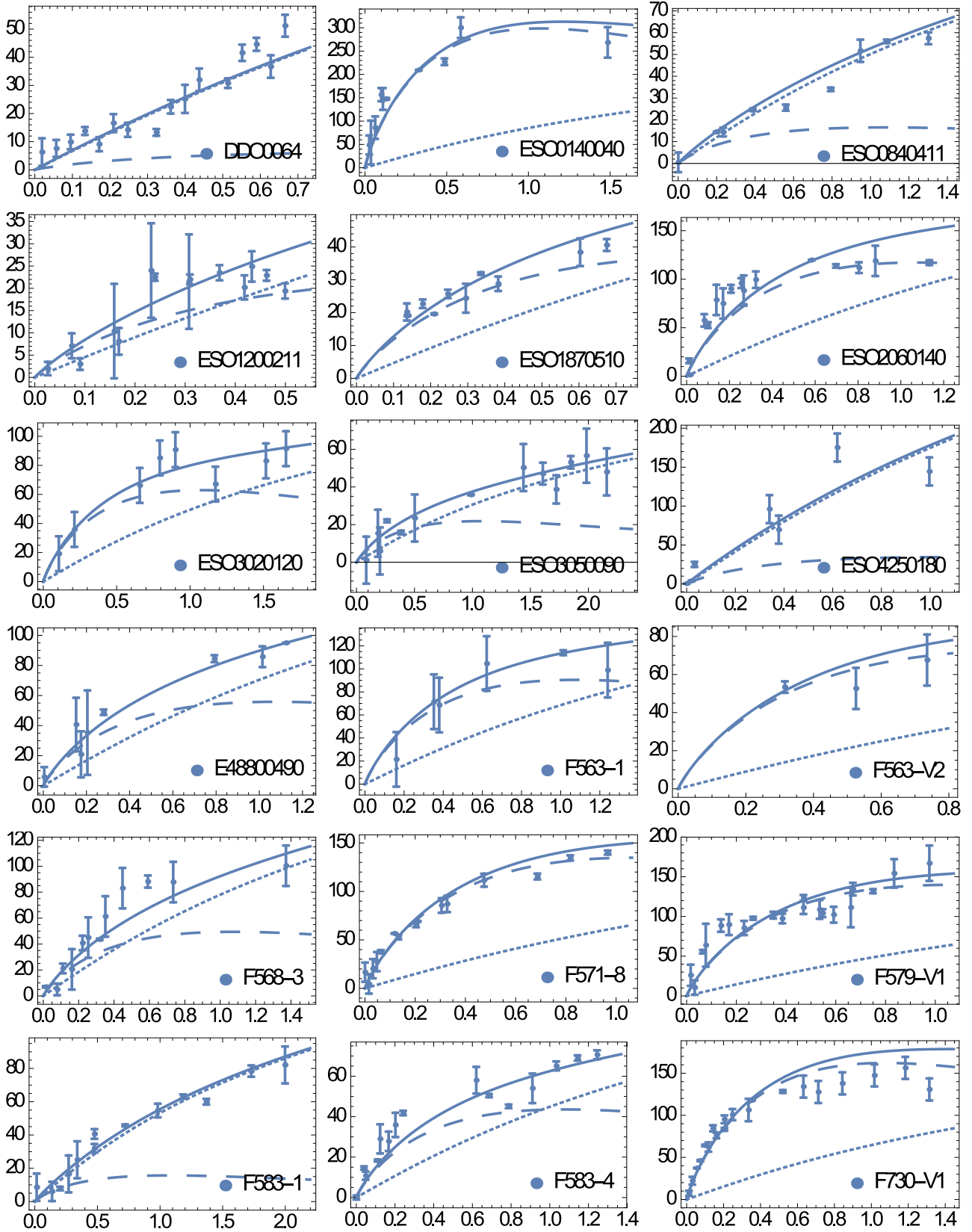}    \\
  \caption{ Fitting of the galactic rotation curves $v$ in $km/s$, with their error bars, as a function of the scaled radius $r/(2R_{d})$ for a sample of 18 galaxies. The contribution due to the luminous Newtonian term alone is shown with  large dashes while the contribution of the new term $2c_{1}a_{0}$ is shown with small dashes. Also, the solid line shows the result of the nonlinear fitting of Eq. \eqref{Rcurve} and the points with the error bars are from the observation. }
  \label{fig:rotfit1}
\end{figure*}

It is worthy to note that the derived mass in the present analysis encompasses all of the matter content of the galaxies including their gas and dust. Reading  the amount  of Hydrogen mass from literature, for example, one can multiply it by $4/3$, inferred from big bang nucleosynthesis  to include the amount of Helium, and derive the total mass of the gas. Usually a small fraction of galactic mass is due  to the mass of primordial Hydrogen and Helium. 

To evaluate the viability of any model, one prefers the rotation curve data to be  extended at least ten times as the scale of the galaxies $R_{d}$. Unfortunately this is not the case for most of the cases and, as we see in Fig \ref{fig:rotfit2}, there are some galaxies with data expanded just about $2R_{d}$ or $R_{d}$ or even less. According to the discussion that we had after Eq. \eqref{Rcurve1}, the rising, falling or constancy of the rotation curve will be clearly observed only after $r \approx 2R_{d}$. Therefore, it is not possible to decide about the accuracy of the model, with a high confident, unless we receive data at larger distances.

Another problem occurs when there are just a few number of data points and they are spread far from each other. For example check F563-V2 or ESO 4250180 in Fig. \ref{fig:rotfit1}. Although in these cases the fit can explain the general behavior of the galactic rotation curve, we should not expect it to follow rotation curve in any single detail. The problem here lies within the mathematical method of fitting which is based on numerical optimization of functions using differential evolution. The more that we have data points, the more accurate the machine can solve the problem  of minimizing and so we can find a better fitting.

\begin{figure*}[!hbp]
\centering
\includegraphics[height=23cm,width=17cm]{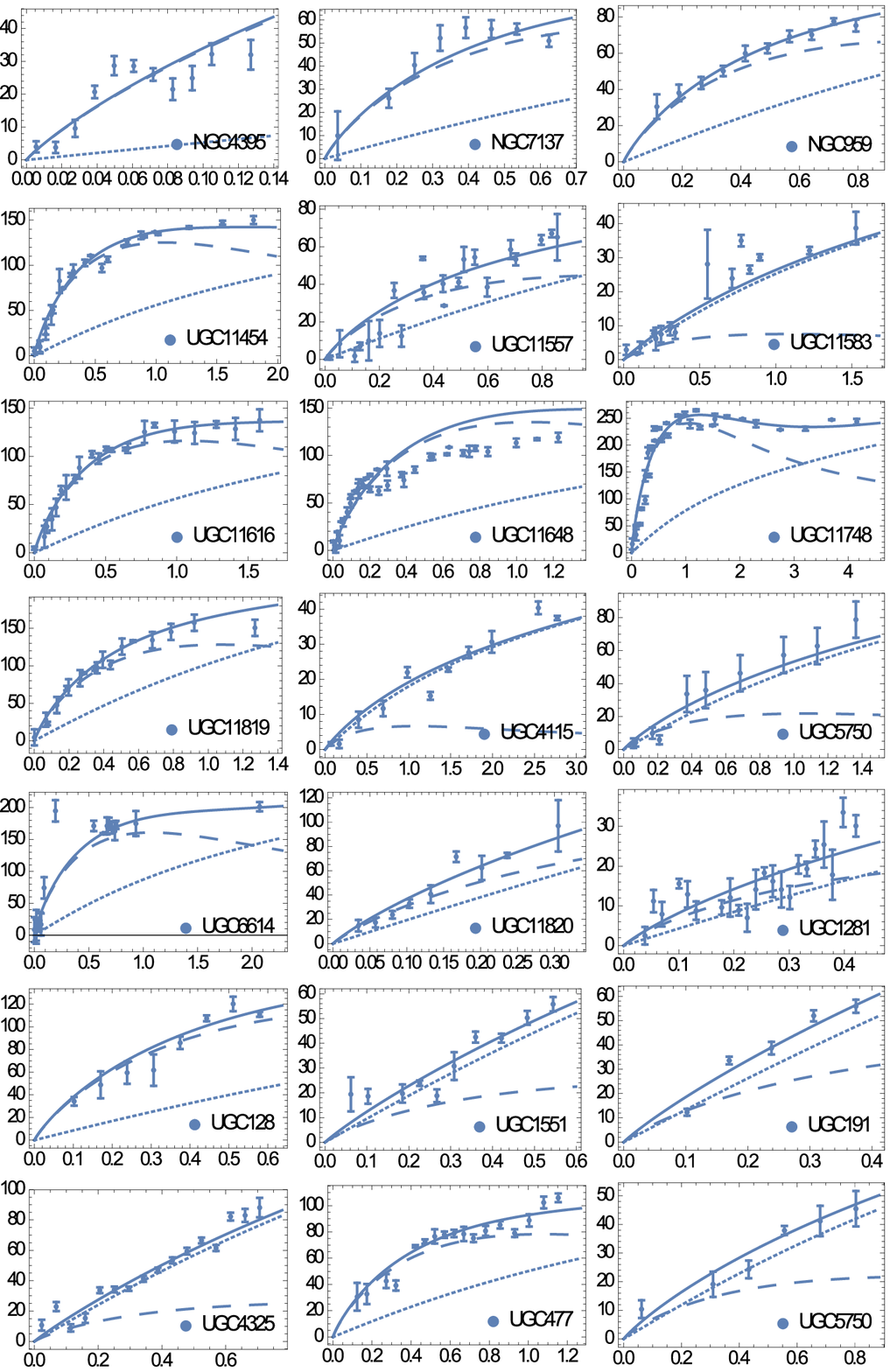}    \\
  \caption{Fitting of the galactic rotation curves $v$ in $km/s$, with their error bars, as a function of the scaled radius $r/(2R_{d})$ for a sample of 21 galaxies.}
  \label{fig:rotfit2}
\end{figure*}

Anyway, from Figs. \ref{fig:rotfit1} and \ref{fig:rotfit2} it is clear that, although we have not considered some of the details of galactic structures, like the presence of a bulge and the effect of disk thickness, our model can capture the general trend of the rotation curves of the present sample. In these figures, the contribution of the luminous Newtonian term alone is specified with  large dashes while the contribution of the new term $2c_{1}a_{0}$ is shown with small dashes and the solid line shows the result of the nonlinear fitting of Eq. \eqref{Rcurve}. The points with the error bars are the observational data.

However, we should point out that for three galaxies out of 39 - namely ESO 2060140, F730-V1 and UGC 11648 - the fitting curve lies above the data and so we observe some difficulties. For ESO 2060140 and UGC 11648, the curves show a magnitude roughly 25 $\%$ larger than the last data points while for the case of  F730-V1 the situation is less serious. In the cases of F730-V1 and UGC 11648 we observed that if we reduce the scale of the galaxies by about 35 $\%$  we can capture the data much better,  though this does not happen for ESO 2060140. See Fig. \ref{fig:rotfit3} for the rotation curves of these three galaxies assuming new galactic scale to be 0.65 times of that in Table \ref{sample}.   This might be a sign of a two-disk or bulge-disk scenario, because it shows that mass is concentrated in smaller radii, or even an overestimation of the galactic scale $R_{d}$.

\begin{figure*}[!hbp]
\centering
\includegraphics[height=3.5cm,width=17cm]{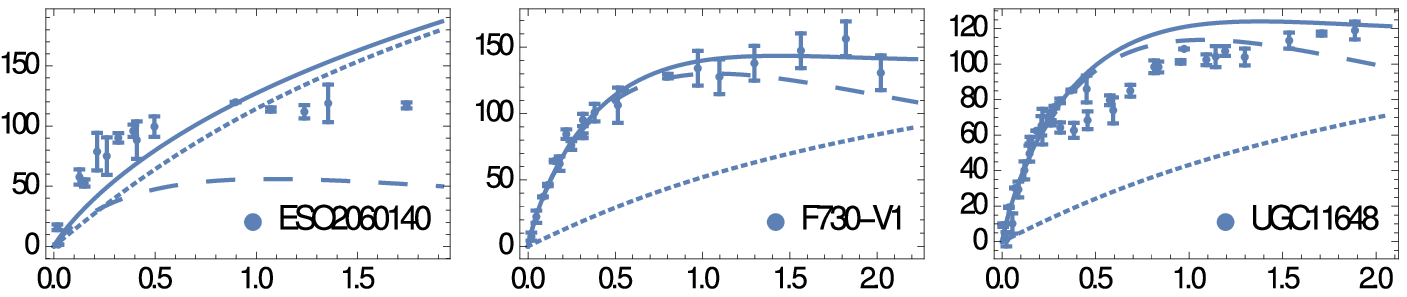}    \\
  \caption{  Fitting of the galactic rotation curves $v$ in $km/s$, with their error bars as a function of the scaled radius $r/(2R_{d})$ for three galaxies which show some difficulties. Here it is assumed that the  scale of the galaxies are 35 $\%$ smaller than the reported values in Table \ref{sample}.  This might be a sign of a two-disk or bulge-disk scenario. By this assumption  F730-V1 and UGC 11648 show a better fit while ESO2060140 is still problematic.}
  \label{fig:rotfit3}
\end{figure*}

We see clear flat rotation curves in the cases of UGC 11454,  UGC 11748 and   UGC 11616 while most of others seem to have rising rotation curves. Although as we argued before, for the galaxies with data spread bellow $r/(2R_{d}) \approx 1$ there isn't enough data to decide about the general trend of the rotation curve with confidence. In addition, we should mention that there is no falling rotation curve in this sample because there is no very massive HSB galaxies. Including this class of galaxies - HSB galaxies - in rotation curve analysis remains as one of our main priorities for future works.

Population synthesis model \citep{Sanders1} suggests an upper limit of mass to light ratio about $10 M_{\odot}/L_{odot}$. From mass to light ratio column in Table \ref{sample} we see that there is only one galaxy, i.e. F571-8, for which we have found a very high $M/L$ equal to about 31 though the fitting curve for this galaxy is quite good. For galaxies with small inclination angle one can argue that the total mass might be overestimated as \cite{mannheim6} have argued about UGC 5999 with an inclination about $14^{\circ}$.  However,  the present case is nearly edge-on and thus the above argument is not applicable. On the hand there are uncertainties in photometric data because of the optical depth and projection effects. Our derived mass is about $5.89 \times 10^{10}M_{\odot}$ which is due to the galaxy's very high reported velocities. Therefore we have a clear problematic case here. Considering the uncertainty in photometric data in addition to the success of this model for the other galaxies it is possible that the luminosity might have been underestimated in this case.

 It should be pointed out that to discuss the dynamics at cluster scales, or probably even huge galaxies with extended rotation curves, one needs to include the cosmological constant term $\Lambda$.  In such cases a new term of the form $1/3 \Lambda c^{2}r$, $r$ being the distance from the center of the galaxy, should be added to the equation of motion  which has the effect of reducing rotation velocity. If we use the current value of cosmological constant $ \Lambda \approx 10^{-52} m^{-2}$, it is possible to see that this term would be important beyond 200 kpc.  Fig \ref{fig:Lambda} explains this point for the case of  $ESO0140040$. As you see, the rotation curve with the contribution of $\Lambda$, shown by solid line, will fall ultimately; otherwise, it will increase indefinitely as shown by  small dashes. In the case of fourth-order conformal gravity too, there is a term similar to $\Lambda$ which causes the ultimate fall off of rotation velocity  \citep{mannheim6}. Though, this happens at smaller radii, compared to our model, because  \cite{mannheim6} have assumed that the term responsible for the fall-off is of the order of $ (1/100 Mpc)^{2} \approx 10^{-50} m^{-2} $. The role of cosmological constant should be studied more  thoroughly in future works. 
\begin{figure*}[!hbp]
\centering
\includegraphics[height=4cm,width=8cm]{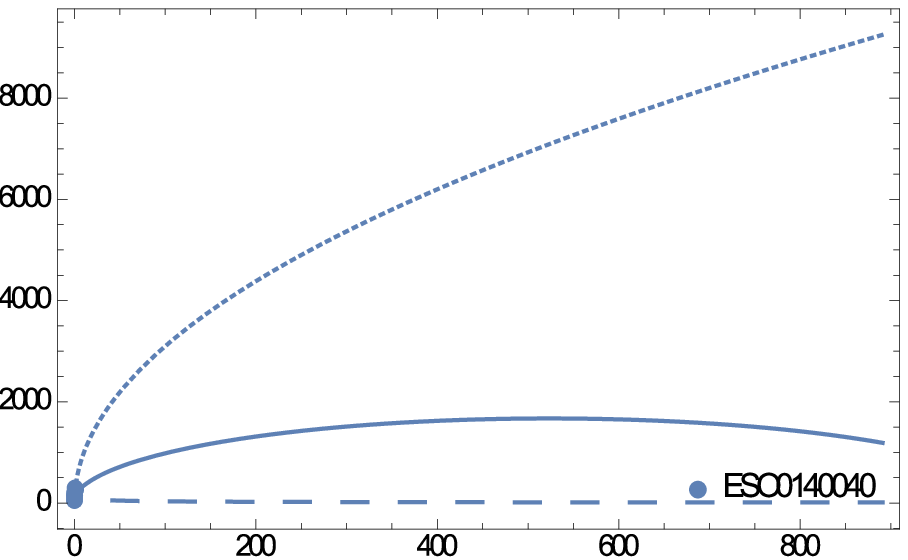}    \\
  \caption{ Rotation curve  $v$ in $km/s$ as a function of the scaled radius $r/(2R_{d})$ for galaxy $ESO0140040$. The contribution due to the luminous Newtonian term alone is shown with  large dashes while the contribution of the new term $2c_{1}a_{0}$ is shown with small dashes. The solid line shows the result of the nonlinear fitting of Eq. \eqref{Rcurve} plus the cosmological term $1/3 \Lambda c^{2}R^{2}_{d}y^{2}$.  }
  \label{fig:Lambda}
\end{figure*}

\section{Conclusion}

In this paper I started to investigate the consequences of a new equation of motion in weak field limit which is motivated by considering Neumann boundary condition and spacetime measurement in an expanding universe. The new term in Eq. \eqref{main} is an intermediate effect between mere curved spacetime - which we use to describe strong fields around massive objects like black holes - and a flat FRW spacetime. We surveyed the consequences of the new term  $ 2c_{1}a_{0}\hat{e_{r}}$ in solar and galactic scales and showed that it leads to some successful results. 

In the case of clusters, external probes  like  lensing can be very helpful in providing a better understanding of the general features of theses objects while interior probes, like X-ray kinematics and dynamics of single galaxies in a cluster, can equip us with  more detailed knowledge.  The same is true about satellite galaxies which  can tell us more about the viability of the present model and specially the role of the $\Lambda$ term.

Regarding the uncertainties in distances of galaxies which are not very high with our current methods - usually within a factor of two at most - and also the uncertainties in velocity data - which are at  most about 20 $\%$ to 30 $\%$ - we see that the universality of acceleration at the last data point should be  regarded as an important dynamical property of the galaxies. We discussed this matter in \cite{ShenavarI} and showed the near universality of the acceleration at the last data points for LSB and HSB galaxies. There we found that $c_{1LSB} = 0.060$ for LSB galaxies and $c_{1HSB} = 0.088$ for HSB galaxies. In the sample of HSB galaxies, there are a few galaxies with  very high luminosities that show a value of Neumann constant about $10 c_{1}$ or more. In these galaxies, the Newtonian term in the equation of motion is completely dominating the constant term $ 2c_{1}a_{0}\hat{e_{r}}$; thus we do not see a constant acceleration yet. However, excluding these data can help to find a better value of $c_{1}$  much closer to the value derived from Friedmann $c_{1}=0.065$.   

One needs to appreciate the fact that galaxies have a variety of shapes and internal structures which, to achieve  better rotation curve fittings, are needed to be included. Among these structures, the most important effect is due to  the presence of bulge. The thickness of galaxies, in some cases,  might be important too. We try to include these factors in future data analysis.

One might worry about the appearance of a new parameter $c_{1}$ in the present work, though, it should be noted that all models that try to solve the mass discrepancy problem, both DM models and modified gravities,  have their new parameters. This is of course natural because, after all, we see a discrepancy between our Newtonian predictions and the observations which should be addressed by one or more new parameters. Of course, we prefer theories with less parameters. 
For example, if we accept the dark matter hypothesis, we will need more parameters to go beyond the standard model of particle physics. For instance, supersymmetry could predict WIMPs, though this theory adds some new parameters to the standard model beside the main point  that the existence of  WIMPs is needed to be confirmed in detectors. In astrophysics of dark matter too, we need some new parameters to describe galactic/extragalactic effects of dark matter, two of which are mentioned in the NFW profile before. In this sense, the present model seems quite economical because it only possesses a single new parameter $c_{1}$. 

In MOND model too, we have a fundamental constant $a_{0}$ with  various proposed interpolating functions $\mu$. For now, it is inferred from MOND literature that the value $a_{0} \approx 1.2 \times 10^{-10} m s^{-2} $ is accepted among them \citep{mond}. This value is close to the value of our new term in the equation of motion $2c_{1}a_{0}= 8.6 \times 10^{-11}$.

From a theoretical point of view, there still remain many other works to do. For example,  according to Newtonian theory of gravity,  galaxies and clusters are unstable without a dark matter halo \citep{stable1,stable2}. This is in contrast with our repeated observations of  galaxies and clusters  which seem to be quite stable. However, numerical simulations have already proved that a linear potential is capable of providing stable disks without needing further matter or dark matter \citep{Christodoulou}, we will check the stability  - and other numerical features - of disk galaxies in a detailed numerical work in future. 

Another important conspiracy is the mass estimation of galactic systems. For example, it is expected that for a cold disk of stars the total mass and virial velocity be related through the following equation: $ v^{3}\sim M $. Instead, according to successive observations, these two parameters are related through the famous Tully-Fisher law: $ v^{4}\sim M $. See \citet{TF1,TF2,TF3} for more details. Any successful modified dynamics/gravity should provide solid proof for this empirical fact and its zero point.  Furthermore, the origin of exponential radial profiles in disk galaxies should be explained. These matters will be covered in another paper \citep{ShenavarIII}.

\appendix
\section{Force Due to Some Geometrically Symmetrical Configurations and Their Virial Energy} \label{app}
In this section we study the force due to our modified equation of motion \eqref{main} using Eq. \eqref{Force} for some useful configurations of matter. In addition, we derive the virial energy of these distributions from the next equation:
\begin{equation}    \label{virial}
Vir = -\int d^{3}\vec{x} \rho(\vec{x}) \vec{x}. \nabla \Phi
\end{equation}
for future reference.

$ \boldsymbol{ {\it Spherical ~Shell:}}$ Consider a point particle with mass $ m $ located at a distance $ r $  from  center of a shell with total mass of $ M $ and radius $ R $. In addition, suppose that $ \rho $ and $ t $  are the density and thickness of this shell. Now, by using the new term in the force \eqref{Force} , i.e. $ \vec{F}^{a_{0}} $, which we introduced before we have:
\begin{eqnarray}     
\vec{F}^{a_{0}} = m  \frac{2c_{1}a_{0}}{M} \int \frac{\vec{x^{\prime}} - \vec{x}}{| \vec{x^{\prime}} - \vec{x}|} \rho(\vec{x^{\prime}})d^{3}\vec{x^{\prime}}
\end{eqnarray}
it is possible to derive the inserted force on mass $ m $ outside of this shell as following:
\begin{equation}
F^{a_{0}}_{out} = \frac{2c_{1}ma_{0}}{M}\frac{\pi t \rho R }{r^{2}} \int^{r+R}_{r-R} (r^{2}-R^{2} +x^{2})dx
\end{equation}   
or $ F^{a_{0}}_{out}= 2c_{1}ma_{0}(1-\frac{R^{2}}{3r^{2}}) $. Similarly, for a point particle inside the shell one has $ F^{a_{0}}_{in} = \frac{4c_{1} m a_{0}r}{3R}$.  The derivation is the same, but the lower limit of the integral is now $ R-r $ instead of $ r-R $. Note that these two forms of force coincide at $ r = R $ and the net force vanishes at $ r=0 $ as it is expected.  Finally, it is possible to use \eqref{virial} and derive the virial energy due to  the $a_{0}$-dependent term  for this configuration as $ Vir_{a_{0}}=-2c_{1}MRa_{0} $.  

$ \boldsymbol{{\it Isotropic ~Sphere:}} $ Consider a sphere with an isotropic mass distribution $ \rho = \rho(x) $ and total mass $ M $ and radius $ R $.  Again, by using the term that is dependent to $a_{0}$ in Eq. \eqref{Force} one may find the force $F^{a_{0}}$  on a point particle $ m $ inside the sphere located at a distance $ r $ from the center of the sphere:
\begin{equation}       
F^{a_{0}} = \frac{2c_{1}m a_{0}}{M} \int \rho(x) \frac{(r-x\cos \theta)x^{2}dx d\phi d\cos \theta}{\sqrt{r^{2}+x^{2}-2rx\cos \theta}} 
\end{equation}
Here we have used the spherical symmetry of the system and neglected non-radial components of the force because of the symmetry. Starting with  integration on $ \theta $, while the integral on $ \phi $ is trivial, one can derive the force $F^{a_{0}}$ exerted on a particle inside sphere: 
\begin{eqnarray}  \label{Force1}    \nonumber
F^{a_{0}}=& \frac{4 \pi c_{1}m a_{0}}{M} \int^{R}_{0} \rho(x) xdx [r+x -  \left| r-x\right|    \\   \nonumber
-&  \frac{r^3+x^3-\left(r^2+r x+x^2\right) \left| r-x\right|  }{3 r^2} ] 
\end{eqnarray}

$ \boldsymbol{{\it Homogeneous ~Sphere:}} $ For a homogeneous sphere, i.e.  $ \rho = \rho_{0} $, by careful evaluation of the last integral  we find $ F^{a_{0}}_{in}= 2c_{1}ma_{0}(r/R - 1/5(r/R)^{3}) $ for a particle inside the sphere  and $ F^{a_{0}}_{out} = 2c_{1}ma_{0} (1-1/5(R/r)^{2}) $ for a particle outside of the sphere. Again, the force vanishes at $ r=0 $ , approaches to $2c_{1}ma_{0} $ at large radii and coincides at $ r=R $.   Now it is easy to find the virial energy of this sphere 
\begin{equation}     \label{Virial4}
Vir_{a_{0}} = \int_{0}^{R} \rho x\frac{F^{a_{0}}_{in}}{m} d^{3}x 
\end{equation}
which we obtain $ Vir_{a_{0}} = -36/35Mc_{1}a_{0} R $. For the case of Newtonian potential energy one obtains $ Vir_{N} = - \frac{3GM^{2}}{5R} $

$ \boldsymbol{{\it Exponential ~Sphere}:}$ Similarly, for an exponential isotropic spherical mass distribution,   one could find the $a_{0}$-dependent force from \eqref{Force1}.  Setting $ \rho = \rho_{0}\exp(-r/R_{d}) $ and $ R \rightarrow \infty  $ into equation \eqref{Force1} it is possible to find the following expression for the force:
\begin{eqnarray}
F^{a_{0}} = 2c_{1}ma_{0} ( 1-\frac{4}{u^{2}} \\   \nonumber
+& \exp(-u)\left[ -\frac{1}{2} -\frac{u^{2}}{6} +4\frac{u+1}{u^{2}} \right] ) 
\end{eqnarray}
which we have used $ u =r/R_{d} $ for more simplicity. Here $R_{d}$is the sphere scale length. To evaluate this equation one needs to use following integral:
\begin{eqnarray}
\int x^{m}e^{cx}dx= (\frac{\partial}{\partial c})^{m}\frac{e^{cx}}{c} 
=\frac{e^{cx}}{c}\Sigma^{m}_{k=0} (-1)^{k}\frac{D^{(k)}x^{m}}{c^{k}}
\end{eqnarray}  
where $ D^{(k)}$ is the $ k^{th}$ derivative operator. Please notify that this force vanishes at $ r=0 $. To prove this one needs to use $ -1/2 = \lim_{u \rightarrow 0} \frac{(u+1)\exp(-u)-1}{u^{2}}$. In addition, it is easy to check that at large radii one has $ F^{a_{0}}=2c_{1}ma_{0}$.  It is also possible to find this force in terms of modified Bessel functions. See \cite{mannheim5} for more details. Although here we prefer to work with exponential function to manifestly explain the finite behavior of the force at very small and large radii. As the last step we calculate the virial energy by putting this force into equation \eqref{Virial4}. Numerical evaluation of this integral is easy and we find $ Vir_{a_{0}}=-(7/2)MR_{d}c_{1}a_{0} $. By evaluating Newtonian virial energy for this mass distribution one could find $ Vir_{N}=-\frac{5GM^{2}}{32R_{d}} $.
Thus, the total virial energy of an exponential sphere becomes:
\begin{equation}
Vir = -\frac{5GM^{2}}{32R_{d}} - \frac{7}{2} c_{1}a_{0} M R_{d}
\end{equation}

$ \boldsymbol{{\it Exponential~ Disk}:}$ Now consider  an exponential disk with an exponential surface mass density:
\begin{equation} 
\Sigma(r) = \Sigma_{0} \exp(-r/R_{d})
\end{equation}
in which $R_{d}$ is the disk scale length. We already know that the Newtonian  potential energy is $ Vir_{N} = -11.63G \Sigma^{2}_{0}R^{3}_{d} $. See \cite{Binney}. To derive the negative virial energy due to the constant acceleration, one should first derive the relevant force. For a mass distribution in cylindrical coordinate we use Eq. \eqref{Potential} to find the radial force which results in:
\begin{eqnarray}   \nonumber
F^{a_{0}}(\varphi , r,z) =&  
\frac{2c_{1}ma_{0}}{M} \int^{2\pi}_{0} d\phi  \int^{\infty}_{0} dZ \int^{\infty}_{0} RdR \rho(R,Z)   \\   
\times & \frac{r-R\cos(\varphi- \phi)}{ \left(  r^{2}+R^{2}-2rR \cos(\varphi -\phi) +(Z-z)^{2} \right) ^{1/2}}  
\end{eqnarray}
If the mass distribution is cylindrically symmetric, like the situation that we have here, then the $\varphi $-component of force vanishes or equivalently we can consider $ \varphi =0 $. Putting Bessel function expansion of Green function in cylindrical coordinate:
\begin{eqnarray}   
\frac{1}{(r^{2}+R^{2}-2rRcos(\varphi -\phi) +(Z-z)^{2})^{1/2}}=~~~~~~~~~~~  \\
\Sigma_{m=-\infty}^{+\infty} \int_{0}^{\infty} dkJ_{m}(kr)J_{m}(kR)\exp[im(\phi-\varphi)-k\mid z-Z\mid]   \nonumber
\end{eqnarray}
into the prior equation and assuming that the mass distribution is  very thin , i.e. $ \rho(R,Z)= \Sigma(R) \delta(Z)$ we find that
\begin{eqnarray}
F^{a_{0}}(r) =   
\frac{4c_{1}\pi ma_{0}}{M} \int^{\infty}_{0} dk  \int^{\infty}_{0} RdR\Sigma(R) \\
\left(  J_{0}(kr)J_{0}(kR) -RJ_{1}(kr)J_{1}(kR) \right)         \nonumber
\end{eqnarray}
where all m-terms have vanished due to  orthogonality conditions for periodic functions except $ m = 0,1$.  Now we should use the following Bessel function integration:
\begin{equation}
\int_{0}^{\infty}RdRJ_{0}(kR)\exp(-\alpha R)= \frac{\alpha}{(\alpha^{2}+k^{2})^{3/2}}
\end{equation}
\begin{equation}
\int_{0}^{\infty}R^{2}dRJ_{1}(kR)\exp(-\alpha R)= \frac{3k\alpha}{(\alpha^{2}+k^{2})^{5/2}}
\end{equation}
\begin{eqnarray}
\int_{0}^{\infty}dk\frac{J_{0}(kr)}{(\alpha^{2}+k^{2})^{3/2}}=&   ~~~~~~~~~~~~~~\\   \nonumber
 \frac{r}{2\alpha}\left( I_{0}(\frac{r\alpha}{2})K_{1}(\frac{r\alpha}{2})-I_{1}(\frac{r\alpha}{2})K_{0}(\frac{r\alpha}{2}) \right) 
\end{eqnarray}
to derive the final result for the $a_{0}$-dependent term:
\begin{eqnarray}    \label{Rcurve3}
 F^{a_{0}} = 4c_{1}ma_{0}y  I_{1}(y)K_{1}(y)  
\end{eqnarray} 
where $y=r/2R_{d} $. In the last step we also have used some relations for differentiating Bessel functions such as $ I^{\prime}_{0}(z)=I_{1}(z) $, $ I^{\prime}_{1}(z)=I_{0}(z)-\frac{I_{1}}{z} $, $ K^{\prime}_{0}(z)=-K_{1}(z) $, $ K^{\prime}_{1}(z)=-K_{0}(z)-\frac{K_{1}}{z} $. This is a relatively tedious but straightforward calculation.  It is also possible, as Mannheim has shown before \cite{mannheim5}, to start from the potential \eqref{Potential} and derive the same force. From the asymptotic expansion of modified Bessel functions:
\begin{eqnarray}  \label{Knu}
K_{\nu}(z) \sim ~~~~~~~~~~~~~~~~~~~~~~~~~~~~~~~~~~~~~~~~~~~~~ \\   \nonumber       
 \sqrt{\frac{\pi e^{-2z}}{2z}}[ 1+\frac{(4\nu^{2}-1^{2})}{1!8z} +\frac{(4\nu^{2}-1^{2})(4\nu^{2}-3^{2})}{2!(8z)^{2}} +\ldots ]
\end{eqnarray} 
\begin{eqnarray} \label{Inu}
I_{\nu}(z) \sim ~~~~~~~~~~~~~~~~~~~~~~~~~~~~~~~~~~~~~~~~~~~~~ \\   \nonumber       
\sqrt{\frac{e^{2z}}{2\pi z}}[( 1-\frac{(4\nu^{2}-1^{2})}{1!8z} +\frac{(4\nu^{2}-1^{2})(4\nu^{2}-3^{2})}{2!(8z)^{2}} +\ldots ]
\end{eqnarray}
it is easy to prove that $ F^{a_{0}}$ vanishes at $ r=0$ and approaches to $ ma_{0} $ in large radii.

 Now we put this force into the definition of virial energy \eqref{virial} to obtain:
\begin{equation}
Vir_{a_{0}} = 64 \pi c_{1}a_{0} \Sigma_{0} R^{3}_{d} \int^{+ \infty }_{0} \exp(-2y) y^{3} I_{1}(y)K_{1}(y)dy  
\end{equation} 
 It is easy to evaluate this last integral numerically. The answer is  0.092. Therefore, the virial energy of an exponential disk becomes:
\begin{equation}
Vir = -11.63/4 \pi^{2} \frac{GM^{2}}{R_{d}} - 2.94c_{1} a_{0} M R_{d}
\end{equation}   
The results for virial energy will be used in future.

\acknowledgments

I would like to thank the authors S. McGaugh, P. D. Mannheim, J. G. O'Brien W.J.G. de Blok,  V. Rubin, K de Naray, A. Bosma,  for providing their data freely. My special thank is to S. McGaugh who provided an easily available galactic database at https://www.astro.umd.edu/~ssm/data/.  Also, I thank Mohammad Moghadassi who helped me in  some graphs. I acknowledge  the anonymous reviewer whose comments helped to clarify the manuscript and especially improved comparisons with other theories. This research has made use of NASA's Astrophysics Data System.

\end{document}